\documentclass[12pt,a4paper]{amsart}

\usepackage{amsmath,amsfonts,amssymb}

\usepackage[hmargin=2cm,vmargin=2cm]{geometry}

\usepackage[hyperfootnotes=false,colorlinks=true,linkcolor=blue,%
citecolor=purple,filecolor=magenta,urlcolor=cyan]{hyperref}%[hidelinks]{hyperref}
\usepackage{nameref,zref-xr}                    

\usepackage{amsthm,amscd}
\usepackage{epsfig}
\usepackage{color}
\usepackage[all]{xy}
\usepackage{tikz}
\usepackage{graphics, setspace}
\usepackage{breqn}
\usepackage{amstext}
\usepackage{array}

\newcommand{\CC}{{\mathbb{C}}}

\newcommand{\QQ}{{\mathbb{Q}}}

\def\A{{\mathcal A}}
\def\B{{\mathcal B}}

\def\I{{\mathcal I}}

\def\res{{\mathrm{res}}}

\newcommand{\p}{{\partial}}

\DeclareMathOperator{\Aut}{Aut}

\newcommand{\ZZ}{\mathbb{Z}}

\newcommand{\bt}{{\bf t}}

\newtheorem{theorem}{Theorem}[section]
\newtheorem{proposition}[theorem]{Proposition}
\newtheorem{lemma}[theorem]{Lemma}
\newtheorem{corollary}[theorem]{Corollary}

\newtheorem{remark}[theorem]{Remark}

\usepackage{color}

\def\&{\vspace{-5pt}&}

\allowdisplaybreaks

\numberwithin{equation}{section}

\begin{document}

\title{
%	Integrable hierarchies of infinite-dimensional Frobenius manifolds of $A$ and $D$ type
Integrable hierarchies associated to %stable 
infinite families of Frobenius manifolds
% of $A$ and $D$ type
}

\author{Alexey Basalaev}
\address{A. Basalaev:\newline Faculty of Mathematics, National Research University Higher School of Economics, Usacheva str., 6, 119048 Moscow, Russian Federation, and \newline
Center for Advanced Studies, Skolkovo Institute of Science and Technology, Nobelya str., 1, 143026 Moscow, Russian Federation}
\email{a.basalaev@skoltech.ru}

\author{Petr Dunin-Barkowski}
\address{P. Dunin-Barkowski:\newline Faculty of Mathematics, National Research University Higher School of Economics, Usacheva str., 6, 119048 Moscow, Russian Federation, \newline
HSE--Skoltech International Laboratory of Representation Theory and Mathematical Physics, Skolkovo Institute of Science and Technology, Nobelya str., 1, 143026 Moscow, Russian Federation, and \newline Institute for Theoretical and Experimental Physics, Bolshaya Cheryomushkinskaya str., 25, 117218 Moscow, Russian Federation}
\email{ptdunin@hse.ru}

\author{Sergey Natanzon}
\address{S. Natanzon:\newline Faculty of Mathematics, National Research University Higher School of Economics, Usacheva str., 6, 119048 Moscow, Russian Federation, and \newline Institute for Theoretical and Experimental Physics, Bolshaya Cheryomushkinskaya str., 25, 117218 Moscow, Russian Federation}
\email{natanzons@mail.ru}

\date{\today}

 \begin{abstract}
 We propose a new construction of an integrable hierarchy associated to any infinite series of Frobenius manifolds satisfying a certain stabilization condition. We study these hierarchies for Frobenius manifolds associated to $A_N$, $D_N$ and $B_N$ singularities. In the case of $A_N$ Frobenius manifolds our hierarchy turns out to coincide with the KP hierarchy; for $B_N$ Frobenius manifolds it coincides with the BKP hierarchy; and for $D_N$ hierarchy it is a certain reduction of the 2-component BKP hierarchy. As a side product to these results we
 %provide the particular 
 %obtain %certain
 illustrate the
 enumerative meaning of certain coefficients of $A_N$, $D_N$ and $B_N$ Frobenius %manifold
 potentials.
 \end{abstract}
 \maketitle

 \setcounter{tocdepth}{1}
   \tableofcontents

 \section{Introduction}
 The theory of Frobenius manifolds was introduced by B.Dubrovin in the early '90s as a general approach to Gromov-Witten theories and certain quantum field theories. %The data of an 
 An $N$--dimensional Frobenius manifold % is
 %collected by
 %contained in
 can be defined (cf. \cite{D2}) via
 its potential $F_N = F_N(t_1,\dots,t_N)$, s.t.
 \begin{equation*}
 \eta_{\alpha,\beta} := \frac{\p^3 F_N}{\p t_1 \p t_\alpha \p t_\beta}
 \end{equation*}
 is a constant non-degenerate matrix and $F_N$ is subject to the following system of equations called the \textit{WDVV equations}:
 \[
    \sum_{\mu,\nu = 1}^N \frac{\p^3 F_N}{\p t_\alpha \p t_\beta \p t_\mu} \eta^{\mu,\nu} \frac{\p^3 F_N}{\p t_\nu \p t_\gamma \p t_\sigma}
    =
    \sum_{\mu,\nu = 1}^N \frac{\p^3 F_N}{\p t_\alpha \p t_\gamma \p t_\mu} \eta^{\mu,\nu} \frac{\p^3 F_N}{\p t_\nu \p t_\beta \p t_\sigma},
 \]
 that should hold for % every given
 all
 $1\leq \alpha,\beta,\gamma,\sigma \leq N$. Here $ \eta^{\mu,\nu}$ is the inverse matrix for matrix $\eta_{\alpha,\beta}$. Note that this
 %is already assumed by the definition
 definition of Frobenius manifolds already assumes
 that $t_1$ is a special variable. Important examples of Frobenius manifolds come from singularity theory, after the work of K.Saito and M.Saito \cite{S1,S2}. In particular, Frobenius manifolds corresponding to ADE singularities %appeared to
 have polynomial potentials $F_N$.
 %\\
 %\\
 
 The %certain 
 connection between Frobenius manifolds and integrable hierarchies has been observed by many authors in %different 
 various
 ways (cf. \cite{DVV,D2,FGM}). Due to the celebrated Witten conjectures, %the 
 particular interest was attributed to %the 
 Frobenius manifolds of ADE singularities \cite{FSZ,FJR}.
 
 It %has been
 was
 a general idea of B.Dubrovin that Frobenius manifolds could be used as a tool to study %ing
 integrable hierarchies. B.Dubrovin and Y.Zhang proposed in \cite{DZ} a way %how 
 to construct an integrable hierarchy associated to any Frobenius manifold. These integrable hierarchies are now called \textit{Dubrovin-Zhang hierarchies}. 
 %The 
 Dubrovin-Zhang hierarchies of ADE Frobenius manifolds turned out to be equivalent to the corresponding Drinfeld-Sokolov hierarchies (cf. \cite{DLZ}). In particular, the Dubrovin-Zhang hierarchy of $A_N$-singularity Frobenius manifold appeared to be equivalent to the $(N+1)$-reduction of the KP hierarchy and the Dubrovin-Zhang hierarchy of $D_N$-singularity Frobenius manifold appeared to be equivalent to the $(2N-2,2)$-reduction of the 2-component BKP hierarchy (cf. \cite{LWZ}).
% \\
% \\
%  We propose a new way how to construct a integrable hierarchy attributed to not just one Frobenius manifold, but a series of Frobenius manifolds, satisfying the certain stabilization conditions. The hierarchy we define has a simple Cauchy-type form following \cite[Lemma~2]{DN}
%  \[
%     \p_\alpha \p_\beta f = \sum_{m\ge 1} \sum_{\gamma_1,\dots,\gamma_m} R_{\alpha,\beta; \gamma_1,\dots,\gamma_m} \p_1\p_{\gamma_1} f\cdots \p_1\p_{\gamma_m}f,
%  \]
%  for the function $f = f(t_1,t_2,\dots)$ depending on the infinite number of variables $t_{k \ge 1}$ and $R_{\alpha,\beta; \bullet}$ the rational coefficients, that we derive from the series of Frobenius manifolds potentials. We show that the compatibility condition of this systems of PDE's is equivalent to the WDVV equation.

We propose a new way 
%how 
to construct an integrable hierarchy 
%attributed 
associated to an infinite series 
% not just one Frobenius manifold, but a series of Frobenius manifolds, 
of Frobenius manifolds (instead of just a single one, as in Dubrovin-Zhang case)
satisfying certain stabilization conditions. The hierarchy we define has a rather simple form.

Fix a collection of  %complex 
numbers $R_{\alpha,\beta; \gamma_1,\dots,\gamma_m} \in \CC$ for $m,\,\alpha,\,\beta,\,\gamma_i \in \mathbb{Z}_{\geq 1}$ (for brevity we will sometimes write $R_\bullet$ for these coefficients in what follows), s.t. they are symmetric w.r.t. interchanging $\alpha$ and $\beta$, and under all permutations of $\gamma_i$'s.
%For 
Consider
an analytic function $f = f(\bt)$ depending on an infinite number of variables $\bt = (t_1,t_2,\dots)$ and denote  $\p_\alpha := \p/\p t_\alpha$ for any $\alpha \in \mathbb{Z}_{\geq 1}$. Consider the following system of PDEs: 
\begin{equation}\label{eq: int sys}
  \p_\alpha \p_\beta f = \sum_{m\ge 1}\sum_{\gamma_1,\dots,\gamma_m \ge 1} R_{\alpha,\beta; \gamma_1,\dots,\gamma_m } \p_1\p_{\gamma_1}f \cdot \dots \cdot \p_1\p_{\gamma_m}f.
\end{equation}
%where we assume that coefficients$R_{\alpha,\beta; \gamma_1,\dots,\gamma_m }$ is symme
%for some prefixed system of complex numbers $R_{\alpha,\beta; \gamma_1,\dots,\gamma_m} \in \CC$. 
It %gives an expression of 
expresses
arbitrary second order derivatives of $f$ via the special second order derivatives~$\p_1\p_\alpha f$ %that we assume as a Cauchy data. 
which are going to be used as the Cauchy data.

%It was observed in \cite[Lemma~2]{DN} that the dispersionless KP hierarchy can be written in this way for an 
%appropriate choice of the coefficients $R_\bullet$. 
Lemma 3.2 of \cite{NZ} implies (after setting $\hbar$ to zero) that the dispersionless KP hierarchy can be written in this way for an 
appropriate choice of the coefficients $R_\bullet$ (see section \ref{ssec:KP} for details;  see also  \cite[Lemma~2]{DN}).
In particular, in this case the coefficients $R_\bullet$ satisfy the following condition:
\[
  R_{\alpha,\beta; \gamma_1,\dots,\gamma_m} = 0 \quad \text{unless} \quad \sum_{p=1}^m \gamma_p = \alpha + \beta,
\]
which results in the summation over $m$ and all $\gamma_i$ in Eq~\eqref{eq: int sys} %is
being
finite for every %fixed
given
 pair $\alpha,\beta$.

\subsection{Construction}\label{section: int sys defintion}
Let $\lbrace F_N \rbrace_{N \ge N_{min}}$ be an infinite series of $N$--dimensional Frobenius manifold potentials with %some minimal dimension 
$F_{N_{min}}$.
Assume $F_N \in \CC[[t_1,\dots,t_N]]$.
%For example one could take 
Specific cases of such infinite series are given by e.g. $F_N = F_{A_N}$ with $N_{min} = 1$, and $F_N = F_{D_N}$ with $N_{min}=4$; see below for more details on these specific examples.

%Let us 
The aim is to
find coefficients $R_\bullet$ s.t. 
\begin{enumerate}
 \item[(1)] system of PDEs \eqref{eq: int sys} is compatible,
 \item[(2)] $F_N$ is a solution to Eq.~\eqref{eq: int sys} for $\alpha + \beta \ll N$, 
\end{enumerate}
which is possible if $F_N$'s satisfy a certain stabilization condition, see below. Here $\alpha + \beta \ll N$ stands for ``sufficiently small $\alpha+\beta$ compared to $N$''. More precisely, we can reformulate condition (2) as follows: $\exists \kappa_1,\kappa_0\in\mathbb{Q}$, s.t. $\forall N\; F_N$ is a solution of \eqref{eq: int sys} for all $\alpha$ and $\beta$ satisfying $\alpha + \beta \leq \kappa_1 N +\kappa_0$. 

%If this is possible, the 
If coefficients $R_\bullet$ satisfying (1) and (2) exist, %are found just 
they can be found from the series expansion of $F_N$ as follows. By the definition of a Frobenius %manifold 
potential we have $\p_1\p_\alpha F_N = \sum_{\beta=1}^N \eta_{\alpha,\beta} t_\beta$, where $\eta$ is a flat metric of the Frobenius manifold. Assume that the coodinates $t_\bullet$ are such that the latter sum only consists of one summand for every given $\alpha$ (in most cases the flat coordinates can be chosen in such a way that $\eta$ is antidiagonal, and this condition holds, cf. \cite{D2}). 
Introducing the notation $t_{\overline \alpha} := \p_1\p_\alpha F_N$ and subtituting $f = F_N$ in Eq~\eqref{eq: int sys} we get:
\begin{equation}\label{eq: R via F}
  \p_\alpha \p_\beta F_N = \sum_{m\ge 1}\sum_{\gamma_1,\dots,\gamma_m} R_{\alpha,\beta; \gamma_1,\dots,\gamma_m} \ t_{\overline{\gamma_1}}\cdot \dots \cdot t_{\overline{\gamma_m}},
\end{equation}
for $\alpha+\beta \ll N$ as in condition (2) above.

Now the numbers $R_\bullet$ are read off as the coefficients of series expansions of $\p_\alpha \p_\beta F_N$. 
In particular, it is straightforward to see that 
\begin{equation}\label{eq: R1beta}
R_{1,\beta; \gamma_1,\dots,\gamma_m}  = R_{\beta,1; \gamma_1,\dots,\gamma_m}= \delta_{m,1}\delta_{\gamma_1,\beta}.
\end{equation}

In order for Eq.~\eqref{eq: int sys} to make sense we need $R_\bullet$ to be independent of $N$. Set 
\begin{align}\label{eq:RFder}
    R_{\alpha,\beta; \gamma_1,\dots,\gamma_m} &= 
    \begin{cases}
        \dfrac{1}{m!} \left.\dfrac{\p^{m+2} F_N}{\p t_\alpha\p t_\beta \p t_{\bar\gamma_1}\cdots \p t_{\bar\gamma_m}} \right|_{\bt = 0} &\text{ if it is independent of } N,
        \\
        0 & \text{ otherwise}.
    \end{cases}
\end{align}
The choice above amounts to a certain stabilization condition on $\p_\alpha\p_\beta F_N$ that should hold after the change of the variables $s_\alpha = t_{\bar\alpha}$ and also the certain choice of allowed indices $\alpha,\beta,\gamma_\bullet$. For such numbers $R_\bullet$ we show in Proposition~\ref{prop: compatibility from WDVV} that  the compatibility condition of system~\eqref{eq: int sys} follows from the WDVV equation on $F_N$, and thus we get a new dispersionless hierarchy.
\begin{remark}\label{rem:Rvanishing}
	Note that such a stabilization condition implies that for fixed $\alpha$ and $\beta$ the sum over $\gamma_1,\dots,\gamma_m$ in the RHS of \eqref{eq: int sys} becomes finite (for a fixed $m$), i.e. $R_{\alpha,\beta; \gamma_1,\dots,\gamma_m}$ are all equal to zero starting with sufficiently large $\gamma_i$'s. In particular, for $N$ sufficiently large for $\p_\alpha\p_\beta F_N$ to stabilize (for given $\alpha$ and $\beta$), if any of the $\gamma_i$ is larger than $N$, $R_{\alpha,\beta; \gamma_1,\dots,\gamma_m}$ necessarily vanishes.
\end{remark}

In the examples of $F_N = F_{A_N}$ and $F_N = F_{B_N}$ such a stabilization condition is just $\alpha+\beta \le N+1$. In the case of $F_N = F_{D_N}$ the stabilization condition is $\alpha+\beta \le N$ and at most one of $\gamma_1,\dots,\gamma_m$ is equal to $N$.%contain not more than one index $N$.

\subsection{Main results}
For the Frobenius manifolds of $A_N$ type the dispersionless hierarchy we construct coincides with the dispersionless KP hierarchy, rather than its reduction as it is in Dubrovin-Zhang case. It's important to note that this coincidence is proved in a easy and straightforward way without employing any complicated techniques. Namely, it turns out that our construction provides the Fay-identities form of the KP hierarchy. 

 For the Frobenius manifolds of $B_N$ type the dispersionless hierarchy we construct coincides with the dispersionless BKP hierarchy of \cite{DJKM}. 
 In $D_N$ case we get a reduction of dispersionless 2-component BKP hierarchy in one of the components only. Note again that Dubrovin-Zhang hierarchy coincides with 2-component BKP hierarchy reduced in both components.
 
  Unfortunately our approach is not applicable directly to the Frobenius manifolds of $E_6$, $E_7$ and $E_8$ singularities, 
  %. Essentially this happens Basically because they don't build up an infinite series. 
  since they are not parts of some obvious infinite series of Frobenius manifolds.
%  We hope to recover these Frobenius manifolds as the instances of some other infinite series for which the integrable hierarchy can be introduced.
We hope that such a series (or multiple series) can be introduced, but it's a subject of future investigations. 
%\\
%\\

 On the way to prove the coincidence of the hierarchies explained above, we 
 %had to show 
 obtained
 the following interesting result.
For every 
%positive $i,j,m$ 
$m,\,\alpha,\,\beta,\,\gamma_1,\dots,\gamma_m \in \mathbb{Z}_{\geq 1}$
denote by $\widehat P_{ij}(\gamma_1,\dots,\gamma_m)$ the number of all partitions $i_1,\dots,i_m$ of $i$ and $j_1,\dots,j_m$ of $j$, s.t. $\forall k\; i_k+j_k = \gamma_k + 1$. 

We have for all $\alpha + \beta \le N+1$ and $\kappa + \sigma \le N$:
\begin{align*}
    \left.\frac{\p^{m+2} F_{A_N}}{\p t_\alpha \p t_\beta \p t_{N+1-\gamma_1}\cdots \p t_{N+1-\gamma_m}} \right|_{\bt=0}
    &= (-1)^{m-1} (m-1)! \cdot \widehat P_{\alpha\beta}(\gamma_1,\dots,\gamma_m),
    \\
    \left.\frac{\p^{m+2} F_{D_N}}{\p t_\kappa \p t_\sigma \p t_{N-\gamma_1}\cdots \p t_{N-\gamma_m}}  \right|_{\bt=0}
    &= (-1)^{m-1} (m-1)!\cdot \widehat P_{2\kappa-1,2\sigma-1}(2\gamma_1-1,\dots,2\gamma_m-1).
\end{align*}
These equalities provide an enumerative meaning of the respective coefficients of $F_{A_N}$ and $F_{D_N}$ potentials.

 \subsection{\texorpdfstring{$\hbar$}{h}--deformation}
 It was observed in \cite{NZ} that the full KP hierarchy can be obtained from the dispersionless KP hierarchy written in the form of system~\eqref{eq: int sys} by the substitution $\p_k \mapsto \p_k^{\hbar}$ for certain differential operators $\p_k^\hbar = \p_k + O(\hbar)$. 
 This phenomenon was investigated deeply in the works of Takasaki and Takebe \cite{TT1,TT2}.
 
 We hope to extend our dispersionless hierarchies to the full form in the same way. In order to do this we need to consider the $\hbar$--deformations of the Frobenius manifold potentials $F_N$, called the higher genera potentials, with the help of Virasoro constraints. We hope to do this in subsequent works.%later.
%\\
%\\
\subsection{Organization of the paper} 
Section \ref{sec:gencompat} is devoted to proving a key result stating that the WDVV equation implies the compatibility of system of PDEs \eqref{eq: int sys} when the coefficients are coming from an infinite family of Frobenius potentials satisfying a stabilization condition.

In section \ref{sec:frobstruc} we recall the basic theory of Frobenius manifolds; then we recall the results due to Noumi-Yamada and Zuber on the form of Frobenius potentials associated to $A_N$, $D_N$ and $B_N$ singularities.

In section \ref{sec:stabpot} we prove that Frobenius potentials associated to $A_N$, $D_N$ and $B_N$ singularities satisfy stabilization conditions of section \ref{section: int sys defintion}.

In section \ref{section: enumuerative meaning} we touch upon the enumerative meaning of the coefficients of these potentials. 

Section \ref{sec:ABhier} is devoted to identifying the hierarchies resulting from the construction of section \ref{section: int sys defintion} applied for $A_N$ and $B_N$ Frobenius potentials with known integrable hierarchies.

Section \ref{sec:Dhier} covers the same subject as section \ref{sec:ABhier}, just for the $D_N$ case, as it turns out that it is quite different to the $A_N$ and $B_N$ ones.
\subsection*{Acknowledgements}
The authors are grateful to A.Buryak for useful discussions.
The work of A.B. was supported by RSF grant no. 19-71-00086.

\section{WDVV and compatibility of a system of PDEs}\label{sec:gencompat}
Assume that $M$ is an open full--dimensional subspace of $\CC^N$. We say that it's endowed with a structure of Frobenius manifold if there is a function $F_N = F_N(t_1,\dots,t_N)$, s.t. the following conditions hold (cf. \cite{D2}).

\begin{itemize}
	\item 
The variable $t_1$ is special in the following sense:
\[
    \frac{\p F_N}{\p t_1} = \frac{1}{2} \sum_{\alpha,\beta = 1}^N \eta_{\alpha,\beta} t_\alpha t_\beta,
\]
%moreover, 
where $\eta_{\alpha,\beta}$ are components of a non-degenerate bilinear form $\eta$ (which does not depend on $t$'s). In what follows denote by $\eta^{\alpha,\beta}$ the components of $\eta^{-1}$.

\item The function $F_N$ %is
%subject to 
satisfies
a large system of PDEs called the WDVV equations:
\[
\sum_{\mu,\nu = 1}^N \frac{\p^3 F_N}{\p t_\alpha \p t_\beta \p t_\mu} \eta^{\mu,\nu} \frac{\p^3 F_N}{\p t_\nu \p t_\gamma \p t_\sigma}
=
\sum_{\mu,\nu = 1}^N \frac{\p^3 F_N}{\p t_\alpha \p t_\gamma \p t_\mu} \eta^{\mu,\nu} \frac{\p^3 F_N}{\p t_\nu \p t_\beta \p t_\sigma},
\]
which should hold for every given $1\leq\alpha,\beta,\gamma,\sigma\leq N$.

\item There is a vector field $E$ called the \textit{Euler vector field}, s.t. modulo quadratic terms in $t_\bullet$ we have $E \cdot F_N = (3-\delta)F_N$ for some fixed complex number $\delta$. We will assume $E$ to have the following simple form
\[
    E = \sum_{k=1}^N d_k t_k \frac{\p}{\p t_k}
\]
for some fixed numbers $d_1,\dots,d_N$. Moreover we set $d_1 = 1$.
\end{itemize}

Given such a data $(M,F_N,E)$ one can endow every tangent space $T_pM$ with a structure of commutative associative product $\circ$ (depending on $\bt$) defined as follows:
\[
    \frac{\p}{\p t_\alpha} \circ \frac{\p}{\p t_\beta} = \sum_{\delta,\gamma=1}^N\frac{\p^3 F_N}{\p t_\alpha \p t_\beta \p t_\delta} \eta^{\delta\gamma} \frac{\p}{\p t_\delta}.
\]
It follows that $\eta(a \circ b,c) = \eta(a,b \circ c)$ for any vector fields $a,b,c$.

The following proposition is very important for what follows.%the logics of this paper.
\begin{proposition}\label{prop: compatibility from WDVV}
    Let the coefficients $R_\bullet$ be constructed as in Section~\ref{section: int sys defintion} from a series of Frobenius manifold potentials $\lbrace F_N \rbrace_{N \ge N_{min}}$ which satisfy a stabilization condition. Then the system of PDEs \eqref{eq: int sys} with such coefficients $R_\bullet$ is compatible.
\end{proposition}
\begin{proof}
  %In order to show that this system of PDEs is compatible we need to prove that $\p_\gamma(\p_\alpha\p_\beta f) = \p_\beta (\p_\alpha \p_\gamma f)$. %
  We need to show that equalities $\p_\gamma(\p_\alpha\p_\beta f) = \p_\beta (\p_\alpha \p_\gamma f)$ hold true if one 
  %rewrites
  substitutes
  the expressions inside the brackets 
  %by
  with the RHS of \eqref{eq: int sys}.% for $f=F_N$.
  
  If one of the indices is equal to $1$ this follows from Eq.~\eqref{eq: R1beta}. 
  
  For any %fixed
  given
  $\alpha,\beta,\gamma \ge 2$ we have% consider:
  \begin{align} \label{eq:tripder}
    \p_\gamma (\p_\alpha\p_\beta f) &= \sum_{k,l\ge 1}\sum_{i_1,\dots,i_k} R_{\alpha,\beta; i_1,\dots,i_k} \sum_{p=1}^k \sum_{q=1}^l \sum_{j_1,\dots,j_l} R_{i_p,\gamma;j_1,\dots,j_l} 
    \\ \nonumber
    & %\quad \cdot \p_1\p_{i_1}f \cdot \dots \cdot\widehat{p}\cdot\dots\cdot \p_1\p_{i_k}f \ \cdot \ \p_1\p_{j_1}f \cdot \dots \cdot \widehat{q} \cdot\dots\cdot \p_1\p_{j_l}f \cdot \p_1\p_1\p_{j_q}f.    
    \quad \times \prod_{s\neq p}\p_1\p_{i_s}f \; \cdot \; \prod_{r\neq q} \p_1\p_{j_r}f \; \cdot \; \p_1\p_1\p_{j_q}f.
  \end{align}
  Here we have applied \eqref{eq: int sys} twice.
  %The coefficients $R_\bullet$ in the %summation written 
  %latter sum
  All coefficients $R_\bullet$ here
  are either zero or can be recovered as coefficients in front of respective monomials in $F_N$ for sufficiently large $N$. %For any given $q$ %fix
  %consider
 % a set $\I := \{i_1,\dots,\widehat{p},\dots,i_k\} \sqcup \{j_1,\dots,\widehat{q},\dots,j_l\}$ giving a non-vanishing contribution to the sum above. Denote %Set
 For a given $K\in \mathbb{Z}_{\geq 0}$, $\I 
 =(\iota_1,\dots\iota_K) 
 \in \left(\mathbb{Z}_{\geq 1}\right)^K$ and $\gamma,\alpha, \beta,\kappa\in \mathbb{Z}_{\geq 1}$ denote
  \begin{align}\label{eq:Omegatildedef}
\widetilde{\Omega}_{\gamma; \alpha,\beta, \I,\kappa} &:= \sum_{h=0}^K 
%\sum_{I_1\sqcup I_2 = \I }\,
%\sum_{\substack{\tilde I_1 = I_1 \sqcup \sigma \\ \tilde I_2 = I_2 \sqcup \kappa }}
\sum_{p=1}^{h+1}\sum_{q=1}^{K-h+1}\sum_{\nu =1}^\infty 
R_{\alpha,\beta;%\tilde I_1}
 \iota_1,\dots,\iota_{p-1},\nu,\iota_p,\dots,\iota_h}
R_{\nu,\gamma;%\tilde I_2}.%
\iota_{h+1},\dots,\iota_{h+q-1},\kappa,\iota_{h+q},\dots,\iota_K,\kappa}\\
&\phantom{:}= \sum_{h=0}^K
\sum_{\nu =1}^\infty (h+1)(K-h+1)
R_{\alpha,\beta;%\tilde I_1}
	\iota_1,\dots,\iota_h,\nu}
R_{\sigma,\gamma;%\tilde I_2}.%
	\iota_{h+1},\dots,\iota_K,\kappa}.%j_q},
\end{align}
where we have used the symmetry of $R_\bullet$ in the second equality.
%  \begin{align}\label{eq:Omegadef}
 %   \Omega_{\gamma; \alpha,\beta, \I,\kappa} &:= %\sum_{s=0}^K 
  %  \sum_{I_1\sqcup I_2 = \I }\,
   % \sum_{\sigma =1}^\infty \sum_{\substack{\tilde I_1 = I_1 \sqcup \sigma \\ \tilde I_2 = I_2 \sqcup \kappa }} R_{\alpha,\beta;\tilde I_1}
    % \iota_1,\dots,\iota_s,\sigma}
%    R_{\sigma,\gamma;\tilde I_2}.%\iota_{s+1},\dots,\iota_K,\kappa}.%j_q}.
 % \end{align}
  
  %It contributes to $\p_\gamma(\p_\alpha\p_\beta f)$ as a summand $\Omega_{\gamma; \alpha,\beta, \I} \cdot \prod_{a \in \I} \p_1\p_{a}f \cdot \p_1\p_1\p_{j_q}f$.
  We can rewrite \eqref{eq:tripder} as
  \begin{equation}\label{eq:tripder2}
  \p_\gamma (\p_\alpha\p_\beta f) = \sum_{K\geq 0}\; \sum_{\I \in \left(\mathbb{Z}_{\geq 1}\right)^K}\; \sum_{\kappa \geq 1} \widetilde{\Omega}_{\gamma; \alpha,\beta, \I,\kappa} \prod_{a \in \I} \p_1\p_{a}f \cdot \p_1\p_1\p_{\kappa}f.
  \end{equation}
  
  Now for a given $K\in\mathbb{Z}_{\geq 0}$,  $\mathcal{J}=(j_1,\dots,j_K)\in\left(\mathbb{Z}_{\geq 1}\right)^K$, $j_1\leq\ldots\leq j_K$ and $\gamma,\alpha, \beta,\kappa\in \mathbb{Z}_{\geq 1}$ denote
   \begin{align}\label{eq:Omegadef}
\Omega_{\gamma; \alpha,\beta, \mathcal{J},\kappa} &:= %\dfrac{1}{|\Aut(\mathcal{J})|}\,
\sum_{\sigma\in S_K}\widetilde{\Omega}_{\gamma; \alpha,\beta, \sigma(\mathcal{J}),\kappa}\\ \nonumber
&\phantom{:}= %\dfrac{1}{|\Aut(\mathcal{J})|}\,
\sum_{\sigma\in S_K}\sum_{h=0}^K
  \sum_{\nu =1}^\infty (h+1)(K-h+1)
  R_{\alpha,\beta;%\tilde I_1}
  	j_{\sigma(1)},\dots,j_{\sigma(h)},\nu}
  R_{\nu,\gamma;%\tilde I_2}.%
  	j_{\sigma(h+1)},\dots,j_{\sigma(K)},\kappa}\\ \nonumber
  &\phantom{:}= \sum_{I_1\sqcup I_2 = \{1,\dots,K\}}
\sum_{\nu =1}^\infty (|I_1|+1)!(|I_2+1)!
R_{\alpha,\beta;%\tilde I_1}
\mathcal{J}_{I_1},\nu}
R_{\nu,\gamma;%\tilde I_2}.%
\mathcal{J}_{I_2},\kappa}.%j_q},
  \end{align}
  
  With the help of this definition of $\Omega_{\gamma; \alpha,\beta, \mathcal{J},\kappa}$, we rewrite \eqref{eq:tripder2} as
  \begin{equation}\label{eq:tripder3}
  \p_\gamma (\p_\alpha\p_\beta f) = \sum_{K\geq 0}\; \sum_{\substack{\mathcal{J}=(j_1,\dots,j_K)\\1\leq j_1\leq\ldots\leq j_K}}\; \sum_{\kappa \geq 1} \dfrac{\Omega_{\gamma; \alpha,\beta, \mathcal{J},\kappa}}{|\Aut(\mathcal{J})|} \prod_{a \in \mathcal{J}} \p_1\p_{a}f \cdot \p_1\p_1\p_{\kappa}f.
  \end{equation}

  It remains to show that each $\Omega_{\gamma; \alpha,\beta, \I,\kappa}$ is symmetric in $\beta$ and $\gamma$.
  
  Let $N_1$ be s.t. for our fixed $\alpha$ and $\beta$ the stabilization condition of Section~\ref{section: int sys defintion} holds for all $R_{\alpha,\beta; \bullet}$ appearing in \eqref{eq:Omegadef}. Due to remark \ref{rem:Rvanishing}, the sum over $\nu$ in \eqref{eq:Omegadef} is actually finite, and \eqref{eq:Omegadef} can be rewritten as
  \begin{equation}\label{eq:Omegafin1}
  \Omega_{\gamma; \alpha,\beta, \I,\kappa} =\sum_{I_1\sqcup I_2 = \{1,\dots,K\}}
  \sum_{\nu =1}^{N_1} (|I_1|+1)!(|I_2+1)!
  R_{\alpha,\beta;%\tilde I_1}
  	\mathcal{J}_{I_1},\nu}
  R_{\nu,\gamma;%\tilde I_2}.%
  	\mathcal{J}_{I_2},\kappa}.
  \end{equation}
  Now let $N_2 > N_1$ be s.t. the stabilization condition of Section~\ref{section: int sys defintion} holds for all $R_{\nu,\gamma; \bullet}$ appearing in \eqref{eq:Omegafin1} for our fixed $\gamma$ and all $1\leq\nu \leq N_1$. Since the terms with $\nu>N_1$ vanish in any case, we can write \eqref{eq:Omegafin1} as follows:
  \begin{equation}\label{eq:Omegafin}
  \Omega_{\gamma; \alpha,\beta, \I,\kappa} =\sum_{I_1\sqcup I_2 = \{1,\dots,K\}}
  \sum_{\nu =1}^{N_2} (|I_1|+1)!(|I_2+1)!
  R_{\alpha,\beta;%\tilde I_1}
  	\mathcal{J}_{I_1},\nu}
  R_{\nu,\gamma;%\tilde I_2}.%
  	\mathcal{J}_{I_2},\kappa},
  \end{equation}
  where all $R_\bullet$'s which are parts of non-vanishing terms satisfy the stabilization condition.

  The following expression is 
  %known to be 
  symmetric in $\beta$ and $\gamma$ 
  %for every fixed $\kappa$
  due to the WDVV equation for $F_{N_2}$:% Expand it using Eq.~\eqref{eq: R via F}.
  \begin{align*}
    & \left.\frac{\p^{|\mathcal{J}|}}{\prod_{a \in \mathcal{J}} \p t_{\bar a}} \left( \sum_{\mu,\nu = 1}^{N_2} \sum_{\sigma  =1}^{N_2} \frac{\p^3 F_{N_2}}{\p t_\alpha \p t_\beta \p t_\mu} \eta^{\mu,\nu} \frac{\p^3 F_{N_2}}{\p t_\nu \p t_\gamma \p t_\sigma} \ \eta^{\sigma,\kappa} \right)\right|_{\bt=0}
    \\
    & = \left.\sum_{I_1\sqcup I_2 = \{1,\dots,K\}}\, \sum_{\nu = 1}^{N_2}  \left( \frac{\p^{|I_1|}}{\prod_{a \in \mathcal{J}_{I_1}} \p t_{\bar a}} \, \frac{\p^3 F_{N_2}}{\p t_\alpha \p t_\beta \p t_{\bar\nu}} \right)\right|_{\bt=0} \left.\left( \frac{\p^{|I_2|}}{\prod_{a \in \mathcal{J}_{I_2}} \p t_{\bar a}} \, \frac{\p^3 F_{N_2}}{\p t_\nu \p t_\gamma \p t_{\bar\kappa}} \right)\right|_{\bt=0}.
  \end{align*}
  The latter form of this expression due to Eq.~\eqref{eq: R via F}, cf. \eqref{eq:RFder}, explicitly coincides with the RHS of \eqref{eq:Omegafin} and thus with $\Omega_{\gamma; \alpha,\beta, \mathcal{J},\kappa}$ which implies that $\Omega_{\gamma; \alpha,\beta, \mathcal{J},\kappa}$ is symmetric in $\beta,\,\gamma$.% after setting $\kappa = j_q$,.
\end{proof}

In the next section we show how one can obtain the data above starting from the associative commutative product $\circ$ and a pairing $\eta$ in the case of A and D singularities.
% In order to make the summation finer assume the potential $F_N$ to be quasi-homogeneous w.r.t. the Euler field $E$:
% \[
%   (3-\delta)F_N = \sum_{k=1}^N d_k t_k \frac{\p F_N}{\p t_k}, \quad \text{ for some fixed } d_1,d_2,\dots,d_N,\delta \in \CC.
% \]
% Then we get the following restriction on the coefficient $R_\bullet$:
% \[
%   R_{\alpha,\beta; \gamma_1,\dots,\gamma_m} = 0, \ \text{unless} \ 3 - \delta - d_\alpha - d_\beta = \sum_{p=1}^N (2-\delta - d_{\gamma_p}).
% \]

\section{Frobenius structures of \texorpdfstring{$A_N$}{AN}, \texorpdfstring{$D_N$}{DN} and \texorpdfstring{$B_N$}{BN} singularities} \label{sec:frobstruc}

The $A_N$ and $D_N$ type singularities are defined %by
via the following polynomials:
\[
    f_{A_N} = \frac{x^{N+1}}{N+1} +y^2 , \quad f_{D_N} = \frac{x^{N-1}}{N-1} + xy^2.
\]
One associates to them the so-called unfoldings $\Lambda_W: \CC^2\times\CC^N \to \CC$
\[
    \Lambda_{A_N} = \frac{x^{N+1}}{N+1} +y^2 + \sum_{k=1}^Nv_k x^{k-1}, \quad \Lambda_{D_N} = \frac{x^{N-1}}{N-1} + xy^2 + \sum_{k=1}^{N-1}v_k x^{k-1} + v_N y,
\]
that depend on additional parameters $v = (v_1,\dots,v_N)\in M_W := \CC^N$. %varying in $M_W := \CC^N$. 

Let's introduce the Frobenius manifold structure on $M_W$. To do this for every fixed $v \in M_W$ consider the following quotient-ring:
\[
    \A_v := \CC[x,y]\, \Big/ \left( \frac{\p \Lambda_W}{\p x}, \frac{\p \Lambda_W}{\p y} \right).
\]
It's endowed with the quotient-ring product structure, and the classical singularity theory arguments assure that $\A_v$ is an $N$--dimensional $\CC$--vector space. Let $c_{ab}^s(v)$ stand for the structure constants of this product in the basis $[\p \Lambda_W/\p v_1],\dots,[\p \Lambda_W/\p v_N]$, namely,
\begin{align*}
    & A_N: \quad \frac{\p \Lambda_W}{\p v_k} = x^{k-1}, \ 1 \le k \le N,
    \\
    & D_N: \quad \frac{\p \Lambda_W}{\p v_k} = x^{k-1}, \ 1 \le k \le N-1, \quad \frac{\p \Lambda_W}{\p v_N} = y.
\end{align*}

The product $\circ: T_vM \otimes T_vM \to T_vM \otimes \CC[v_1,\dots,v_N]$ is now defined by
\[
    \frac{\p}{\p v_a} \circ \frac{\p}{\p v_b} := \sum_{k=1}^N c_{ab}^k(v) \frac{\p}{\p v_k}.
\]
Obviously, $\p/\p v_1$ is the unit of this product.
In particular, we have for $A_N$
\begin{align}
    \frac{\p}{\p v_a} \circ \frac{\p}{\p v_b} = \frac{\p}{\p v_{a+b-1}} \quad \forall a+b \le N+1,
    \label{eq: A_N simple str constants in v}
    \\
    \frac{\p}{\p v_a} \circ \frac{\p}{\p v_{N+2-a}} = - \sum_{k=2}^N (k-1)v_k \frac{\p}{\p v_{k-1}}.
\end{align}

Introduce the non-degenerate $\CC[v]$--bilinear pairing $\eta: T_vM \otimes T_vM \to \CC[v_1,\dots,v_N]$ by
\begin{align*}
    &\eta(\frac{\p}{\p v_a},\frac{\p}{\p v_b}) := c_{a,b}^N, \quad W = A_N,
    \\
    &\eta(\frac{\p}{\p v_a},\frac{\p}{\p v_b}) := c_{a,b}^{N-1}, \quad W = D_N.
\end{align*}
The pairing we introduce is in fact the well-known residue pairing. 

\begin{theorem}[cf. \cite{D1,S1,ST}]
    The data $(M,\circ,\eta)$ is a Frobenius manifold. In particular, there is a choice of the coordinates $t_\alpha = t_\alpha(v)$ (which are called the flat coordinates), s.t. in the basis $\p/\p t_1, \dots, \p/\p t_N$ we have
    \begin{itemize}
     \item the pairing $\eta$ is constant,
     \item there is a Frobenius manifold potential $F_W = F_W(t_1,\dots,t_N)$, s.t.
     \[
        \frac{\p}{\p t_\alpha} \circ \frac{\p}{\p t_\beta} = \sum_{\gamma,\delta=1}^N \frac{\p^3 F_W}{\p t_\alpha \p t_\beta \p t_\gamma} \eta^{\gamma \delta} \frac{\p}{\p t_\delta}.
     \]

    \end{itemize}
\end{theorem}

For the cases of $A_N$ and $D_N$ singularities, the flat coordinates of the theorem above were investigated by Noumi and Yamada in \cite{NY}. We use their result in what follows applying the certain rescaling of the coordinates that makes the formulae simpler. They also gave the formulae for the potentials $F_W$.

\subsection{\texorpdfstring{$A_N$}{AN} and \texorpdfstring{$D_N$}{DN} Frobenius manifold potentials}
Let $W$ be either $A_N$ or $D_N$. The potential $F_W$ is a polynomial in $t_1,\dots,t_N$ with rational coefficients subject to the quasi-homogeneity condition $E_W \cdot F_W = (3 - \delta_W) F_W$ with
\begin{align*}
    & E_{A_N} = \sum_{\alpha=1}^{N}\frac{N+2-\alpha}{N+1} t_\alpha \frac{\p}{\p t_\alpha}, \quad \delta_{A_N} = \frac{N-1}{N+1},
    \\
    & E_{D_N} = \sum_{\alpha=1}^{N-1}\frac{N-\alpha}{N-1} t_\alpha \frac{\p}{\p t_\alpha} + \frac{N}{2(N-1)} t_N \frac{\p}{\p t_N}, \quad \delta_{D_N} = \frac{N-2}{N-1}. 
\end{align*}
The pairing $\eta$ reads
\begin{align*}
    & \eta_{\alpha,\beta} = \delta_{\alpha+\beta,N+1} \quad & \text{ for } W = A_N,
    \\
    & \eta_{\alpha,\beta} = 
    \begin{cases}
    1 \quad \text{when} \quad \alpha = \beta = N,
    \\
    \delta_{\alpha+\beta,N} \quad \text{otherwise}.
    \end{cases} \quad & \text{ for } W = D_N.
\end{align*}
We see that for all these $W$ for any $\alpha \in \{1,\dots,N\}$ there exists a unique integer $\bar\alpha \in \{1,\dots,N\}$ such that $\eta_{\alpha,\bar\alpha}= 1$.
%According to this for each $W$ we define involution $\bar{\phantom{x}}$ %for
%acting on $x \in \{1,\dots,N\}$, s.t. 
%\[
% \bar x := \left\lbrace \alpha \in \{1,\dots,N\} \ | \ \eta_{\alpha,x}= 1 \right\rbrace.
%\]
%Because of the explicit form of $\eta$, this involution is unique for $W$ we assume.
%\\
%\\

For $W = A_N, D_N$ Noumi-Yamada gave the formulae for the potential $F_W(t_1,\dots,t_N)$ in the following way. They introduce functions $\psi^{(r)}_\gamma \in \QQ[v_1,\dots,v_N]$ depending of the unfolding variables $v_k$ as above, s.t. for all $1 \le \alpha \le N$ the following equations hold:
\begin{align*}
    \frac{\p F_W}{\p t^\alpha} &= \psi^{(2)}_{\overline \alpha}(t_1,\dots,t_N),
    \\
    t_\alpha & = \psi_\alpha^{(1)}(v_1,\dots,v_N).
\end{align*}
It is only reasonable to consider the potential $F_W$ in flat coordinates $t_k$ and therefore it is important to invert the above formula of \cite{NY} in order to express $v_k = v_k(t)$.

\subsection{\texorpdfstring{$A_N$}{AN} case}
We have $\overline \alpha = N+1-\alpha$ and
\begin{align*}
    \psi_\gamma^{(r)}(v) &:= \sum_{\substack{\alpha_1,\dots,\alpha_N \ge 0 \\ \sum_{k=1}^{N} (N+2-k) \alpha_k = r(N+1) +1 - \gamma}} \left(-1\right)^{|\alpha|-r} \prod_{k=0}^{|\alpha|-1-r}(\gamma + k(N+1)) \prod _{k=1}^{N} \frac{v_k^{\alpha _k}}{\alpha _k!},
\end{align*}
where $|\alpha| = \sum_{k=1}^{N} \alpha_k$.

The inverted formulae were given by Buryak in \cite{B1} from the study of open Gromov-Witten theories:
\begin{align}\label{eq: An essential coordinate via flat}
     v_\gamma &=  \sum_{\substack{\alpha_1,\dots,\alpha_N \ge 0 \\ \sum_{k=1}^{N} (N+2-k) \alpha_k = N+2 - \gamma}} \frac{(|\alpha|+\gamma-2)!}{(\gamma-1)!} \prod _{k=1}^{N} \frac{t_k^{\alpha _k}}{\alpha _k!}.
\end{align}
Note that the condition $\sum_{k=1}^{N} (N+2-k) \alpha_k = N+2 - \gamma$ precisely ensures that $v_\gamma$ is quasi-homogeneous (w.r.t. the Euler field $E_{A_N}$) and its weight is equal to the weight of $t_\gamma$. %all the summation conditions above are just the weight conditions. In particular, for $v=v(t)$ it gives that $v_\gamma$ has the same weight as $t_\gamma$ does.

\subsection{\texorpdfstring{$D_N$}{DN} case}
We have 
\[
    \overline \alpha = N-\alpha, \ 1\le\alpha \le N-1, \quad \overline N = N.
\]
and
\begin{align*}
    \psi^{(1)}_{\gamma} &= \sum_{\substack{\alpha_1,\dots,\alpha_{N-1} \ge 0 \\ \sum_{k=1}^{N-1} (N-k)\alpha_k = N- \gamma}} \left(-1\right)^{|\alpha|-1} \prod _{k=0}^{|\alpha|-2} (2 \gamma -1 +2 k (N-1)) \prod_{k=1}^{N-1} \frac{v_k^{\alpha_k}}{\alpha_k!}, \quad 1 \le \gamma \le N-1,
    \\
    \psi^{(1)}_{N} &= v_N.
\end{align*}
where $|\alpha| = \sum_{k=1}^{N-1} \alpha_k$.

In order to introduce $\psi^{(2)}_\gamma$ let us define the following combinatorial coefficients:%consider first the following quantities.
\begin{align*}
    A^{(1)}_{\gamma,\alpha} &:= \left(-1\right)^{|\alpha|-2} \prod _{k=0}^{|\alpha|-3} (2\gamma -1 + 2 k (N-1)), 
    \\
    A^{(2)}_{\gamma,\alpha} &:= \left(-1\right)^{|\alpha|-1} \prod _{k=0}^{|\alpha|-2} (2 \gamma -1 + 2 k (N-1)), \qquad 1 \le \gamma \le N-2,
    \\
    A^{(2)}_{N-1,\alpha} &:= 2.
\end{align*}
Then
\begin{align*}
    \psi_\gamma^{(2)}(v) &:= \sum_{\substack{\alpha_1,\dots,\alpha_{N-1} \ge 0 \\ \sum_{k=1}^{N-1} (N-k)\alpha_k = 2(N-1) +1- \gamma}} A^{(1)}_{\gamma,\alpha} \prod _{k=1}^{n-1} \frac{v_k^{\alpha _k}}{\alpha _k!}
     \\
     &
    + \sum_{\substack{\alpha_1,\dots,\alpha_{N-1} \ge 0 \\ \sum_{k=1}^{N-1} (N-k)\alpha_k = N-1-\gamma}} \frac{A^{(2)}_{\gamma,\alpha}}{2} \prod _{k=1}^{n-1} \frac{v_k^{\alpha _k}}{\alpha _k!} \frac{v_N^2}{2}, \qquad  1 \le \gamma \le N-1,
    \\
    \psi_{N}^{(2)}(v) &:= v_1v_N,
\end{align*}

\subsection{\texorpdfstring{$B_N$}{BN} Frobenius manifold potential}
This Frobenius manifold does not correspond to a deformation theory of a hypersurface singularity. This makes its definition more involved. It was shown by Zuber in \cite{Z} that the following equation holds
\begin{equation}\label{eq: Bn via An}
    F_{B_N}(t_1,\dots,t_N) = F_{A_{2N-1}}(t_1,0,t_2,0,t_3,\dots,t_N).
\end{equation}
We %assume 
use
the above equation% above 
as the definition of $F_{B_N}$.
It follows that $\eta_{\alpha,\beta} = \delta_{\alpha+\beta,N+1}$ and
\[
    E_{B_N} = \sum_{\alpha=1}^{N}\frac{N +1 -\alpha}{N} t_\alpha \frac{\p}{\p t_\alpha}, \quad \delta_{B_N} = \frac{N-1}{N}.
\]

\section{Stabilization of \texorpdfstring{$A_N$}{AN}, \texorpdfstring{$B_N$}{BN}, and \texorpdfstring{$D_N$}{DN} potentials}%Structure of $A_N$, $B_N$, and $D_N$ potentials revisited}
\label{sec:stabpot}
In this section we discuss in details the structure of $A_N$, $B_N$ and $D_N$ Frobenius manifold potentials. In particular, we prove the respective stabilization statements in Theorem~\ref{theorem: An stabilization}, Proposition~\ref{prop: Bn stabilization} and Theorem~\ref{theorem: Dn stabilization}.

\subsection{\texorpdfstring{$A_N$}{AN} case}
\begin{theorem}\label{theorem: An stabilization}
    For any $N_2 > N_1 \ge 1$ and $\alpha,\beta$, s.t. $1 \le \alpha,\beta \le N_1$, $\alpha + \beta \le N_1+1$ we have
    \[
        \left.\frac{\p^2 F_{A_{N_1}}}{\p t_\alpha \p t_\beta} \right|_{\forall \gamma\; t_{N_1+1-\gamma} = s_\gamma} = \left.\frac{\p^2 F_{A_{N_2}}}{\p t_\alpha \p t_\beta} \right|_{\forall \gamma\; t_{N_2 + 1- \gamma} = s_\gamma},
    \]
    understood as an equality of polynomials in $s_\bullet$.
\end{theorem}
\begin{proof}
    In this proof we have to make use of both flat coordinates $t_\bullet$ and unfolding coordinates $v_\bullet$.
    Denote by $c_{ab}^r = c_{ab}^r(v)$ the structure constants in the basis $\p / \p v_\bullet$. Consider also the basis change matrices
    \[
        \Psi^\alpha_a := \frac{\p t_\alpha}{\p v_a}, \ \Psi_\alpha^a := \frac{\p v_a}{\p t_\alpha}
    \]
    where we use the greek and latin letters for $t$ and $v$ coordinates repsectively. We have ${\sum_a \Psi^\alpha_a \Psi^a_\beta = \delta^\alpha_\beta}$.
    \begin{lemma}
        The matrices $\Psi^\alpha_a$ and $\Psi_\alpha^a$ stabilize. Namely, for $v_a^{(N_\bullet)}%(\bt) 
        = v_a^{(N_\bullet)}(\bt)$ being the expression of unfolding coordinates via flat coordinates %, computed 
        for $A_{N_\bullet}$ we have
        \[
            \left.\frac{\p v^{(N_1)}_a}{\p t_{\alpha}} \right|_{\forall \gamma\; t_{N_1+1-\gamma} = s_\gamma} = \left.\frac{\p v^{(N_2)}_a}{\p t_{\alpha}} \right|_{\forall \gamma\; t_{N_2+1-\gamma} = s_\gamma}
        \]
        for $1 \le \alpha,a \le N_1$ and $N_1 < N_2$.
    \end{lemma}
    \begin{proof}
        By using Eq.\eqref{eq: An essential coordinate via flat} we have
        \begin{align*}
            \left.\frac{\p v^{(N)}_a}{\p t_\delta} \right|_{t_{N+1-\gamma} = s_\gamma} & = 
            \sum_{\substack{\alpha_1,\dots,\alpha_{N} \ge 0 \\ \sum_{k=1}^{N} (N+2-k) \alpha_k = \delta - a}} \left.\frac{(|\alpha|+ a-1)!}{(a-1)!} \prod _{k=1}^{N} \frac{t_k^{\alpha _k}}{\alpha _k!} \right|_{t_{N+1-\gamma} = s_\gamma}
            \\
            & = \sum_{\substack{\alpha_1,\dots,\alpha_{N} \ge 0 \\ |\alpha| + \sum_{k=1}^{N} k \alpha_k = \delta - a}} \frac{(|\alpha|+ a-1)!}{(a-1)!} \prod _{k=1}^{N} \frac{s_k^{\alpha_k}}{\alpha _k!}
            \\
            & = \sum_{m \ge 0} \frac{1}{m!} \sum_{\alpha_1+\dots + \alpha_m = \delta-a-m} \frac{(m + a-1)!}{(a-1)!} s_{\alpha_1}\cdot\dots\cdot s_{\alpha_m}.
        \end{align*}
        It is now straightforward to see that the last expression we obtained  %have got 
        does not depend on $N$, what concludes the proof.
    \end{proof}
    
    For any $\gamma$ and $N\ge 1$ we have
    \[
        c^{(N)}_{\alpha\beta\gamma} := \frac{\p^3 F_{A_{N}}}{\p t_\alpha \p t_\beta \p t_\gamma} = \sum_{r=1}^N \sum_{a,b=1}^N \Psi_r^{N+1-\gamma}\Psi^a_\alpha\Psi^b_\beta c_{ab}^r.
    \]
    The change of coordinates $t_\alpha = t_\alpha(v)$ is quasi-homogeneous. In particular, it follows that $\Psi^\alpha_a = 0$ unless $\alpha \le a$. Therefore for any $K \le N$ and $\alpha + \beta \le K+1$ we have
    \[
        c^{(N)}_{\alpha,\beta,N+1-\gamma} 
        = \sum_{r=1}^N \sum_{\substack{1 \le a,b \le N \\ a+b \le K+1}} \Psi_r^\gamma\Psi^a_\alpha\Psi^b_\beta c_{ab}^r
        = \sum_{\substack{1 \le a,b \le N \\ a+b \le K+1}} \Psi_{a+b-1}^\gamma\Psi^a_\alpha\Psi^b_\beta,
    \]
    where we have used Eq.\eqref{eq: A_N simple str constants in v} to get the second equality.
    
    By using the lemma above and the quasi-homogeneity of the change of the variables $t_\bullet = t_\bullet(v)$ we have for all $\alpha + \beta \le N_1 + 1$:
    \[
        \left.c^{(N_1)}_{\alpha,\beta,N_1+1-\gamma} \right|_{t_{N_1+1-\gamma} = s_\gamma} = \left.c^{(N_2)}_{\alpha,\beta,N_2+1-\gamma} \right|_{t_{N_2+1-\gamma} = s_\gamma},
    \]
    which concludes the proof of the theorem.
\end{proof}

By using the quasi-homogeneity condition of $F_{A_N}$ we have for any $\alpha+\beta \le N+1$ the equality
\begin{equation}\label{eq: Fan potential}
    \frac{\p^2 F_{A_N}}{\p t_\alpha \p t_\beta} = \sum_{\gamma=1}^N \frac{1+\gamma}{\alpha+\beta} t_{N+1-\gamma} \sum_{a,b=1}^N \frac{\p t^\gamma}{\p v_{a+b-1}} \frac{\p v_a}{\p t_\alpha} \frac{\p t_b}{\p v_{\beta}}.
\end{equation}

\subsection{\texorpdfstring{$B_N$}{BN} case}
The following stabilization proposition is straightforward in proof but nontrivial in its statement.
\begin{proposition}\label{prop: Bn stabilization}
    For any $N_2 > N_1 \ge 1$ and $\alpha,\beta$, s.t. $1 \le \alpha,\beta \le N_1$, $\alpha + \beta \le N_1+1$ we have
    \[
        \left.\frac{\p^2 F_{B_{N_1}}}{\p t_\alpha \p t_\beta} \right|_{\forall \gamma\; t_{N_1+1-\gamma} = s_\gamma} = \left.\frac{\p^2 F_{B_{N_2}}}{\p t_\alpha \p t_\beta} \right|_{\forall \gamma\; t_{N_2+1-\gamma} = s_\gamma},
    \]
    understood as an equality of polynomials in $s_\bullet$.
\end{proposition}
\begin{proof}
    This follows immediately from the definition of $F_{B_N}$ and Theorem~\ref{theorem: An stabilization}.
\end{proof}

The $B_N$ Frobenius manifolds were introduced via the $A_N$ Frobenius manifolds.
In what follows it will be useful to build up the connection of the $B_N$ Frobenius manifolds to the $D_N$ Frobenius manifolds too.

\begin{proposition}\label{proposition: Bn inside Dn}
    We have
    \begin{equation}\label{eq: Bn via Dn}
        F_{B_N}(t_1,\dots,t_N) = F_{D_{N+1}}(t_1,t_2, \dots, t_{N},0).
    \end{equation}
\end{proposition}
\begin{proof}
    For $D_{N+1}$ Frobenius manifold setting $t_{N+1} = 0$ is equivalent to setting $v_{N+1} = 0$. One shows easily that for $A_{2N-1}$ Frobenius manifold setting all $t_{2a} = 0$ is equivalent to setting all $v_{2a} = 0$. 

    We compare the functions $\psi^{(r)}$ for both cases. In this proof we denote by ${}^A\psi_\gamma^{(r)}$ and ${}^D\psi_\gamma^{(r)}$ the respective $\psi$--functions of $A_{2N-1}$ and $D_{N+1}$ respectively.    
    It follows immediately from the definition that we have
    \[
        {}^A\psi_{2a-1}^{(1)}(v_1,0,v_2,0,\dots,v_{2N-1})
        =
        {}^D\psi_{a}^{(1)}(v_1,v_2,\dots,v_{N-1},0)
    \]
    for all $1 \le a \le N-1$.
    
    Comparing the $\psi^{(2)}$--functions we should take care of the involution on both sides. It remains to note that ${}^A\psi^{(2)}_{2N-2b+1} = {}^D\psi^{(2)}_{N+1-b}$ which completes the proof.
\end{proof}

\begin{corollary}
The following formula expresses the dependence of $D_N$ coordinate $v$ on the flat coordinate $t$
    \begin{align}\label{eq: Dn essential coordinate via flat}
     v_b &=  \sum_{\substack{\alpha_1,\dots,\alpha_{N-1} \ge 0 \\ \sum_{k=1}^{N-1} (N-k) \alpha_k = N - b}} \frac{(|\alpha|+2b-3)!}{(2b-2)!} \prod _{k=1}^{N-1} \frac{t_k^{\alpha _k}}{\alpha _k!},
     \\
     v_N &= t_N,
    \end{align}
    where $|\alpha| = \sum_{k=1}^{N-1} \alpha_k$.
\end{corollary}
\begin{proof}
    This follows immediately from the above proposition and formula~\eqref{eq: An essential coordinate via flat}.
\end{proof}

\begin{remark}
    One could also give a geometric proof of the isomorphism of two $N$--dimensional Frobenius submanifolds of $A_{2N-1}$ and $D_{N+1}$. However in what follows we need the flat structures of both submanifolds to agree. Namely we indeed want to fix not only the isomorphism class of both Frobenius submanifolds but the potentials too.
\end{remark}

\begin{remark}
    The construction of Dubrovin attributes a Frobenius manifold to a Weyl group. In that sense Frobenius manifolds theory could not distinguish between $B_N$ and $C_N$ root systems whose Weyl groups coincide. The two submanifolds above are just two ways of how to find the same Weyl group as a subgroup of $A_\bullet$ and $D_\bullet$ Weyl groups. 
    
    Studying the solutions to ``open'' WDVV equation it was found in \cite{BB1} that potentials $F_{A_{2N-1}}$ and $F_{D_{N+1}}$ can be accompanied with the ``open'' potentials $F^o_{A_{2N-1}}$ and $F^o_{D_{N+1}}$ being some new functions of $t_\bullet$ and one additional variable $s$. The open potential $F^o_{A_{2N-1}}$ is a polynomial in $s$ and the open potential $F^o_{D_{N+1}}$ is a Laurent polynomial with an order two pole in $s$. This made us hope that ``open'' theories could distinguish between $B_N$ and $C_N$ root systems. Unfortunately the equality of the above proposition holds for the open potentials too.
\end{remark}

\subsection{\texorpdfstring{$D_N$}{DN} case}\label{section: DN revisited}
For any fixed $N$ denote by $v_1^{(N)}(\bt)$ the polynomial expressing the $v_1$ coordinate of $D_N$ via $t_1,\dots,t_{N-1}$. The formulae of Noumi-Yamada show that
\begin{align}\label{eq: D_N potential explicit in t}
    \frac{\p F_{D_N}}{\p t_\gamma} &= \A_{\gamma}^{(N)} + \B_{\gamma}^{(N)} \cdot t_N ^2,
    \quad 1 \le \gamma \le N-1,
    \\
    \frac{\p F_{D_N}}{\p t_N} &= v_1^{(N)}(\bt) \cdot t_N,
\end{align}
with $\A_{\gamma}^{(N)},\B_{\gamma}^{(N)} \in \QQ[t_1,\dots,t_{N-1}]$. Namely, these functions do not depend on $t_N$.

\begin{theorem}\label{theorem: Dn stabilization}
    For any $N_2 > N_1 \ge 4$ 
%     and $\alpha,\beta$, s.t. $1 \le \alpha,\beta \le N_1$, $\alpha + \beta \le N_1$ 
    we have
    \begin{align*}
        &\left.\frac{\p v_1^{N_1}(\bt)}{\p t_\beta}\right|_{\forall \gamma \; t_{\gamma} = s_{N_1-\gamma}} = \left.\frac{\p v_1^{N_2}(\bt)}{\p t_\beta}\right|_{\forall \gamma \; t_{\gamma} = s_{N_2-\gamma}}, \qquad \forall \beta < \min(N_1,N_2),
        \\
        &\left.\frac{\p \A_{\alpha}^{N_1}(\bt)}{\p t_\beta}\right|_{\forall \gamma \; t_{\gamma} = s_{N_1-\gamma}} = \left.\frac{\p \A_{\alpha}^{N_2}}{\p t_\beta}\right|_{\forall \gamma \; t_{\gamma} = s_{N_2-\gamma}},
         \qquad \forall \alpha+\beta < \min(N_1,N_2)
    \end{align*}
    understood as an equality of polynomials in $s_\bullet$.
\end{theorem}
\begin{proof}
    For a fixed $N$ using Eq.\eqref{eq: Dn essential coordinate via flat} we get
    \begin{align*}
        \frac{\p v_1(\bt)}{\p t_\beta} & = 
        \sum |\alpha|!  \prod _{k=1}^{N-1} \frac{t_k^{\alpha _k}}{\alpha _k!}, \quad 1 \le \beta \le N-1, \quad |\alpha| := \sum_{k=1}^{N-1}\alpha_k,
    \end{align*}
    where the summation is taken over all $\alpha_1,\dots,\alpha_{N-1} \ge 0$ satisfying 
    ${\sum_{k=1}^{N-1} (N-k) \alpha_k = \beta - 1}$. The last equation can be rewritten as ${\sum_{k=1}^{N-1} k \alpha_{N-k} = \beta - 1}$. 
    After taking the involution $\bar \alpha_\bullet := N - \alpha_\bullet$ one notes that the derivatives we compute only depend on $N$ via the number of summands. However for every fixed $\beta$ the numbers $k$, s.t. $k \ge \beta$ only contribute to the solution set with $\alpha_k$ because we should have $\alpha_\bullet \ge 0$. Once the solution set $\{\alpha_\bullet \}$ is obtained for some $N$, it contributes to all the higher ones, but %however
    with %the
     shifted indices. The shift is exactly the $D_N$ involution $\overline x := N-x$.
    
    The second statement follows immediately from Proposition~\ref{proposition: Bn inside Dn} and Theorem~\ref{theorem: An stabilization}.
\end{proof}

%\subsection{Enumerative meaning}\label{section: enumuerative meaning}
\section{Enumerative meaning of the coefficients of \texorpdfstring{$A_N$}{AN} and \texorpdfstring{$D_N$}{DN} potentials}\label{section: enumuerative meaning}
In this section we discuss enumerative meaning of coefficients of the $A$ and $D$ singularities Frobenius manifold potentials. 

Consider the quotient-ring 
\[ 
\A_N := \QQ[\bt]\otimes\QQ[[z_1^{-1},z_2^{-1}]]\, \big/\, (z_1^{-(N+1)}z_1^{-1}, z_1^{-N}z_1^{-2}, \dots, z_1^{-1}z_1^{-(N+1)}).
\]
Namely, this is the finite rank $\QQ[\bt]$--module generated by polynomials in $z_1^{-1}$ and $z_2^{-1}$ with the total degree not exceeding $N+1$.

\begin{proposition}\label{prop: Fay-type for An}
    Denote $\p_\alpha := \p / \p t_\alpha$ and $F = F_{A_N}$.
    In the ring $\A_N$ we have
        \begin{equation}
            (z_1-z_2) \exp\left(\sum_{\alpha,\beta \ge 1} z_1^{-\alpha}z_2^{-\beta} \p_\alpha \p_\beta F \right) = (z_1-z_2) - \left( \sum_{\alpha \ge 1}z_1^{-\alpha}\p_1\p_\alpha F - \sum_{\beta \ge 1}z_2^{-\beta}\p_1\p_\beta F \right).
        \end{equation}
%    In other words, we only make the statement about the second derivatives
Note that this statement only concerns the second derivatives $\p_\alpha \p_\beta F$ s.t. $\alpha+\beta \le N+1$.
\end{proposition}
\begin{proof}
    Because of the special role of variable $t_1$ in a Frobenius manfiold potential, it is easy to see that the desired equality holds in the rank $2N$ submodule $\QQ[\bt]\otimes\QQ\langle z_1^{-1},z_2^{-1},\dots,z_1^{-N},z_2^{-N}\rangle \subset \A_N$.
    
    In what follows we are going to use the expression%another series
    \[
        P := \sum_{\substack{\alpha,\beta = 1 \\ \alpha+\beta \le N+1}}^N (\alpha+\beta)\frac{\p^2 F_{A_N}}{\p t_\alpha \p t_\beta} z_1^{-\alpha}z_2^{-\beta} = - \left(z_1\frac{\p}{\p z_1} + z_2 \frac{\p}{\p z_2} \right) \cdot \sum_{\substack{\alpha,\beta = 1 \\ \alpha+\beta \le N+1}}^N \frac{\p^2 F_{A_N}}{\p t_\alpha \p t_\beta} z_1^{-\alpha}z_2^{-\beta}.
    \]
    Denote $P_{\alpha,\beta} := [z_1^{-\alpha}z_2^{-\beta}] P$. %By employing the 
    %above
    %equality %above 
    %it's 
    With the help of the above equality it is easy to see that 
    the statement of the proposition is equivalent to the following one:
%    equivalent to show that
    \[
        P_{\alpha+1,\beta} - \sum_{\delta=1}^{\beta-1} t_{N+1-\delta} P_{\alpha,\beta-\delta} = P_{\alpha,\beta+1} - \sum_{\delta=1}^{\alpha-1} t_{N+1-\delta} P_{\alpha-\delta,\beta}.
    \]
    In coordinates by using Eq.~\eqref{eq: Fan potential} this is equivalent to
    \begin{align*}
        &\sum_{\gamma=1}^N (N+2-\gamma) t_\gamma \sum_{a,b} \frac{\p t_{N+1-\gamma}}{\p v_{a+b-1}} \left[ \frac{\p v_a}{\p t_{\alpha+1}}\frac{\p v_b}{\p t_\beta}-\frac{\p v_a}{\p t_{\alpha}}\frac{\p v_b}{\p t_{\beta+1}} \right]
        \\
        & \quad = \sum_{\gamma=1}^N (N+2-\gamma) t_\gamma \sum_{a,b} \frac{\p t_{N+1-\gamma}}{\p v_{a+b-1}} \left[
        \sum_\delta t_{N+1-\delta} \frac{\p v_a}{\p t_\alpha}\frac{\p v_b}{\p t_{\beta-\delta}}
        -         \sum_\delta t_{N+1-\delta} \frac{\p v_a}{\p t_{\alpha-\delta}}\frac{\p v_b}{\p t_{\beta}}
        \right],
    \end{align*}
    which should hold for all $\alpha+\beta \le N$. This equality can be checked combinatorially via Eq.\eqref{eq: An essential coordinate via flat}, but we are going to use a more geometrical approach.
    
    Denote by $\Lambda := \Lambda_{A_N}(x,\bt)$ the unfolding of $A_N$ singularity written in the flat coordinates. Set
    \[
        \phi_\alpha := \frac{\p \Lambda}{\p t_\alpha} = \sum_{k=1}^N \frac{\p v_k}{\p t_\alpha} x^{k-1} \in \CC[x]\otimes\CC[t_1,\dots,t_N].
    \]
    These functions satisfy the following recursive relation:
    $ x \phi_\alpha = \phi_{\alpha+1} + \sum_{\sigma=0}^{\alpha-1} t_{N+1+\sigma-\alpha} \phi_\sigma$.
    
    In terms of these functions we have
    \[
        \sum_{a+b=p}\frac{\p v_a}{\p t_\alpha}\frac{\p v_b}{\p t_\beta} = [x^{p-2}] \left( \phi_\alpha \cdot \phi_\beta\right),
    \]
    where the product in the bracket is just $x$--polynomial product. 
    By using this observation and the recursive relations on $\phi_\bullet$ we have
    \begin{align*}
        &\sum_{a+b = p} \left( \frac{\p v_a}{\p t_{\alpha+1}}\frac{\p v_b}{\p t_\beta}-\frac{\p v_a}{\p t_{\alpha}}\frac{\p v_b}{\p t_{\beta+1}} \right) = [x^{p-2}] \left( \phi_{\alpha+1}\phi_\beta - \phi_\alpha\phi_{\beta+1} \right)
        \\
        &= [x^{p-2}] \left( (x\phi_\alpha - \sum_{\sigma=0}^{\alpha-1} t_{N+1-(\alpha-\sigma)}\phi_{\sigma} )\phi_\beta - \phi_\alpha(x \phi_\beta - \sum_{\sigma=0}^{\beta-1} t_{N+1-(\beta-\sigma)}\phi_\sigma) \right)
        \\
        &= [x^{p-2}] \left(- \sum_{\delta=1}^\alpha t_{N+1-\delta}\phi_{\alpha-\delta} \phi_\beta + \sum_{\delta=1}^\beta t_{N+1-\delta} \phi_\alpha\phi_{\beta-\delta} \right)
        \\
        & = \sum_{a+b = p} \left[ \sum_{\delta=1}^\beta t_{N+1-\delta} \frac{\p v_a}{\p t_{\alpha}} \frac{\p v_{b}}{\p t_{\beta-\delta}} - \sum_{\delta =1}^\alpha t_{N+1-\delta} \frac{\p v_a}{\p t_{\alpha-\delta}} \frac{\p v_b}{\p t_\beta}
        \right].
    \end{align*}
    This completes the proof.
\end{proof}

Proposition~\ref{prop: Fay-type for An} allows us to make a statement about the combinatorial meaning of the $F_{A_N}$ coefficients.

For every positive $i,j,m$ denote by $\widehat P_{ij}(\gamma_1,\dots,\gamma_m)$ the number of all partitions $i_1,\dots,i_m$ of $i$ and $j_1,\dots,j_m$ of $j$, s.t. $\forall k \; i_k+j_k=\gamma_k +1$.

\begin{corollary}\label{corollary: An enumuerative}
% Let $s_{k} := t_{N+1-k}$.
For every $\alpha+\beta \le N+1$ and $m\ge 1$ we have
\[
    \frac{\p^{m+2} F_{A_N}}{\p t_\alpha \p t_\beta \p t_{N+1-\gamma_1}\cdots \p t_{N+1-\gamma_m}} = (-1)^{m-1} (m-1)! \cdot \widehat P_{\alpha\beta}(\gamma_1,\dots,\gamma_m).
\]
\end{corollary}
\begin{proof}
    This follows from the proof of Lemma~3.2 in \cite{NZ}. We repeat it here for completeness.
    Recall that $\p_1\p_\alpha F = t_{N+1-\alpha} = s_\alpha$.
    By proposition~\ref{prop: Fay-type for An} we have in $\A_N$ (in the notation of that proposition):
    \begin{align*}
        &\sum_{\alpha,\beta \ge 1} z_1^{-\alpha}z_2^{-\beta} \p_\alpha \p_\beta F = \log\left[ 1 - \frac{\sum_{\alpha \ge 1} (z_1^{-\alpha} - z_2^{-\alpha} )s_\alpha}{z_1-z_2} \right]
        \\
        & \quad = \log\left[ 1 + z_1^{-1}z_2^{-1} \sum_{\alpha \ge 1} \frac{ z_1^{-\alpha} - z_2^{-\alpha}}{z_1^{-1}-z_2^{-1}} s_\alpha \right]
         = \log \left[1+ \sum_{p\ge 1} \left(\sum_{i+j = p + 1} z_1^{-i}z_2^{-j}\right) s_p \right]
        \\
        & \quad = \sum_{m\ge 1} \frac{(-1)^{m-1}}{m} \sum_{i,j\ge 1}z_1^{-i}z_2^{-j} \sum_{\substack{i_1+\dots+i_m = i, \\ j_1+\dots+j_m=j}} s_{i_1+j_1-1}\cdots s_{i_m+j_m-1}.
    \end{align*}
    The rest follows by comparing the coefficients of $z_1^{-\alpha}z_2^{-\beta}$ on both sides of equation.
\end{proof}

\begin{corollary}\label{corollary: Dn enumuerative}
% Let $s_k := t_{N+1-k}$. 
For any $\alpha+\beta \le N$ and $m\ge 1$ we have
    \begin{equation*}
        \frac{\p^{m+2} F_{D_N}}{\p t_\alpha \p t_\beta \p t_{N-\gamma_1}\cdots \p t_{N-\gamma_m}}
%         [t_{N-\gamma_1}\cdots t_{N-\gamma_m}]\frac{\p^{2} F_{D_N}}{\p t_\alpha \p t_\beta }  
        = (-1)^{m-1} (m-1)!\cdot \widehat P_{2\alpha-1,2\beta-1}(2\gamma_1-1,\dots,2\gamma_m-1).
    \end{equation*}
\end{corollary}
\begin{proof}
    This follows immediately from Proposition~\ref{proposition: Bn inside Dn} and the above corollary.% above.
\end{proof}

\section{\texorpdfstring{$A$}{A} and \texorpdfstring{$B$}{B} hierarchies}\label{sec:ABhier}
In this section we present the integrable hierarchies associated to the series of $A_N$ and $B_N$ Frobenius manifolds. Theorems~\ref{theorem: A-hierarchy is KP} and~\ref{theorem: Binfty hierarchy is BKP} beneath show that they coincide with the KP and BKP hierarchies respectively.

\subsection{Dispersionless hierarchy of type \texorpdfstring{$A$}{A}}
% In this section consider $s_\alpha := t_{N+1-\alpha}$.
For any $\alpha,\beta \ge 1$, s.t. $\alpha+\beta \le N+1$ set
 \begin{equation*}
     R^{(A_N)}_{\alpha,\beta; \gamma_1,\dots,\gamma_m} := \left.\frac{1}{m!} \frac{\p^{m+2} F_{A_N}}{\p t_\alpha \p t_\beta \p t_{N+1-\gamma_1}\cdot\dots\cdot \p t_{N+1-\gamma_m}} \right|_{\bt=0}.
 \end{equation*}
% \begin{equation*}
%     R^{(A_N)}_{\alpha,\beta; \gamma_1,\dots,\gamma_m} := \frac{1}{m!}  \p_\alpha \p_\beta \p_{N+1-\gamma_1}\cdot\dots\cdot \p_{N+1-\gamma_m} F_{A_N} \mid_{\bt=0}.
% \end{equation*}
It follows by stabilization Theorem~\ref{theorem: An stabilization} that the following quantities are well-defined 
\[
    R^{\mathrm{A}}_{\alpha,\beta; \gamma_1,\dots,\gamma_m} := R^{(A_{\alpha+\beta+1})}_{\alpha,\beta; \gamma_1,\dots,\gamma_m},
\]
giving us an infinite set of rational numbers.
One notes immediately that $R^{\mathrm{A}}_{\alpha,\beta; \gamma_1,\dots,\gamma_m}$ is only non-zero when $\alpha+\beta -k = \gamma_1+\dots+\gamma_m$, which is essentially the quasi-homogeneity condition.

Consider the infinite system of PDEs on $f = f(t_1,t_2,\dots)$, Eq. \eqref{eq: int sys}:
\begin{equation}\label{eq: Ainfty Cauchy}
    \p_\alpha \p_\beta f = \sum_{m \ge 1} \sum_{\gamma_1 + \dots + \gamma_m = \alpha+\beta - m} R^{\mathrm{A}}_{\alpha,\beta; \gamma_1,\dots,\gamma_m} \p_1 \p_{\gamma_1} f \cdots \p_1 \p_{\gamma_m} f.
\end{equation}
This is a Cauchy-type system of PDEs expressing any second order derivatives of $\p_\alpha\p_\beta f$ via the special second order derivatives $\p_1\p_\bullet f$.

It follows from Proposition~\ref{prop: compatibility from WDVV} that the system of PDEs \eqref{eq: Ainfty Cauchy} is compatible.

\subsection{KP hierarchy} \label{ssec:KP}
We are going to introduce KP hierarchy in a Hirota form as a system of equations on the function $\tau = \tau(t_1,t_2,\dots)$, and also in the form of Fay identities as a system of equations on $F = \hbar^{2}\log(\tau)$. We skip many important details of the general theory that do not play any particular role in our exposition. They can be found for example in \cite{DN,NZ}. 

For any function $\tau = \tau(t_1,t_2,\dots)$ and formal variables $z,\hbar$ denote
\[
    \tau(\bt \pm [z^{-1}]) := e^{\pm \hbar D(z)} \cdot \tau(\bt), \quad D(z) := \sum_{k \ge 1} \frac{z^{-k}}{k} \p_k.
\]
It's straightforward to note that the action of $\exp(D(z))$ is just the change of variables $\lbrace t_k \rbrace_{k=1}^\infty \mapsto \lbrace t_k + \hbar z^{-k}/k \rbrace_{k=1}^\infty$. KP hierarchy in Hirota form is the following equality in the ring of formal power series in $\bt,\bt'$:
% \[
%     (z_1-z_2)\tau^{[z_1,z_2]}\tau^{[z_3]} + (z_2-z_3)\tau^{[z_2,z_3]}\tau^{[z_1]} + (z_3-z_1)\tau^{[z_3,z_1]}\tau^{[z_2]} = 0,
% \]
\begin{align*}
    \res\left( e^{\xi(\bt'-\bt,z)} \tau(\bt'- [z^{-1}]) \tau(\bt+ [z^{-1}]) dz\right) = 0,
\end{align*}
where $\xi(\bt,z) := \sum_{n\ge 1} t_n z^n$.

In what follows we need to consider the KP hierarchy in terms of $F = \hbar^{2}\log(\tau)$.
Consider %one more
another
differential operator
\[
    \Delta(z) := \frac{\exp(\hbar D(z)) - 1}{\hbar} = D(z) + O(\hbar).
\]
Then Hirota equation above is equivalent to the following equation in the ring of formal power series in $z_1^{-1},z_2^{-1}$, called Fay identity:
\[
    \exp(\Delta(z_1)\Delta(z_2) F ) = 1 - \frac{\Delta(z_1)\p_1F - \Delta(z_2)\p_1F}{z_1-z_2}.
\]
The following proposition is crucial for our exposition.    

\begin{proposition}[Lemma~3.2 in \cite{NZ}]\label{prop: Fay by NZ}
%     Let $P_{ij}(\gamma_1,\dots,\gamma_m) := ij \widetilde P_{ij}(\gamma_1,\dots,\gamma_m)$ for $\widetilde P_{i,j}$ as in Section~\ref{section: enumuerative meaning}.
    Fay identities on the function $F = \hbar^2\log \tau$ are equivalent to the following system of equations
    \[
        \p^\hbar_i\p^\hbar_j F = \sum_{m \ge 1} \frac{(-1)^{m-1}}{m} \sum_{\gamma_1+\dots+\gamma_m = i+j-m} \frac{ij}{\gamma_1\cdots\gamma_m}
        \widetilde P_{ij}(\gamma_1,\dots,\gamma_m) \p_1\p_{\gamma_1}^\hbar F \cdots \p_1\p_{\gamma_m}^\hbar F,
    \]
    where $\p_k^\hbar$ are the differential operators defined by the equality
    \[
        \Delta(z) = \sum_{k \ge 1} \frac{z^{-k}}{k} \p_k^\hbar.
    \]
    In particular, we have $\p_k^\hbar = \p_k + O(\hbar)$.
\end{proposition}

Assume $F = \sum_{g \ge 0} F_g \hbar^g$. It follows immediately from the proposition above that if $F$ is subject to Fay identities, then the function $F_0$ satisfies
\[
    \p_i\p_j F_0 = \sum_{m \ge 1} \frac{(-1)^{m-1}}{m} \sum_{\gamma_1+\dots+\gamma_m = i+j-m} \frac{ij}{\gamma_1\cdots\gamma_m}
    \widetilde P_{ij}(\gamma_1,\dots,\gamma_m) \p_1\p_{\gamma_1} F_0 \cdots \p_1\p_{\gamma_m} F_0.
\]
This system of equations is called the dispersionless limit of the KP hierarchy.

\subsection{Identification}

\begin{theorem}\label{theorem: A-hierarchy is KP}
    The system of PDEs \eqref{eq: Ainfty Cauchy} coincides with the dispersionless KP hierarchy after the change of variables $t_k \mapsto t_k/k$.
    
    Full KP hierarchy is obtained from the system of PDEs \eqref{eq: Ainfty Cauchy} via the substitution $\p_k \mapsto \p_k^\hbar$.
\end{theorem}
\begin{proof}
    It follows immediately from Proposition~\ref{prop: Fay-type for An} and Corollary~\ref{corollary: An enumuerative}.
    
    At the same time the full KP hierarchy is obtained from its dispersionless limit via the substitution $\p_k \mapsto \p_k^\hbar$, which completes the proof.  
\end{proof}

\subsection{BKP hierarchy}\label{section: 1-BKP}
This hierarchy was introduced in \cite{DJKM} via the Lax form (see also \cite{N1,N2} for another context). We present it here via the bilinear identity form on the function $\tau  = \tau(\bt)$, for $\bt = \lbrace t_1, t_3, t_5, \dots \rbrace$, following~\cite{T}.

Consider the operators
\[
    D^{\mathrm{B}}(z) := \sum_{n \ge 0} \frac{z^{-2n-1}}{2n+1} \p_{2n+1}, \quad \Delta^{\mathrm{B}}(z) := \frac{\exp(2\hbar D(z))-1}{\hbar}.
\]
Denote 
\[
    \tau(\bt \pm 2[z^{-1}]) := e^{\pm 2 \hbar D^{\mathrm{B}}(z)} \cdot \tau(\bt), \quad \xi^{\mathrm{B}}(\bt,z) := \sum_{n\ge 0} t_{2n+1}z^{2n+1}.
\]

The BKP hierarchy is the following equation
\begin{align*}
    \res\left( e^{\xi^{\mathrm{B}}(\bt'-\bt,z)} \tau(\bt'- 2[z^{-1}]) \tau(\bt+ 2[z^{-1}])\frac{dz}{z}\right) = \tau(\bt') \tau(\bt).
\end{align*}
The corresponding Fay identity for $F = \hbar^2\log\tau$ reads
\begin{align}
    &\left(z_1 + z_2 - \p_1 \hbar \Delta^{\mathrm{B}}(z_1)\Delta^{\mathrm{B}}(z_2) F  - \p_1(\Delta^{\mathrm{B}}(z_1)F + \Delta^{\mathrm{B}}(z_2)F)\right)
    \exp(\Delta^{\mathrm{B}}(z_1)\Delta^{\mathrm{B}}(z_2) F)
    \tag{BKP}\label{bkp}
     \\
     &\quad\quad\quad = \frac{z_1+z_2}{z_1-z_2} \left( z_1-z_2 - \p_1(\Delta^{\mathrm{B}}(z_1)F - \Delta^{\mathrm{B}}(z_2)F) \right).
     \notag
\end{align}
The dispersionless limit is
\begin{align}
    \left( 1 - \frac{\p_1 (2D^{\mathrm{B}}(z_1) + 2D^{\mathrm{B}}(z_2))F_0}{z_1+z_2} \right) e^{2D^{\mathrm{B}}(z_1)\cdot2D^{\mathrm{B}}(z_2) F_0}
     = 1 - \frac{\p_1 (2D^{\mathrm{B}}(z_1) - 2D^{\mathrm{B}}(z_2))F_0}{z_1-z_2}.
\tag{BKP-dl}\label{bkp-dl}
\end{align}

\subsection{Dispersionless hierarchy of type \texorpdfstring{$B$}{B}}
Consider an infinite set of rational numbers
\begin{equation*}
    R^{(B_N)}_{\alpha,\beta; \gamma_1,\dots,\gamma_k} := \left.\frac{\p^{m+2} F_{B_N}}{\p t_\alpha \p t_\beta \p t_{N+1-\gamma_1}\cdot\dots\cdot \p t_{N+1-\gamma_k}} \right|_{\bt=0}.
\end{equation*}
It follows by Proposition~\ref{prop: Bn stabilization} that the following quantities are well-defined 
\[
    R^{\mathrm{B}}_{\alpha,\beta; \gamma_1,\dots,\gamma_k} := R^{(B_{\alpha+\beta+1})}_{\alpha,\beta; \gamma_1,\dots,\gamma_k}.
\]
One notes immediately that $R^{\mathrm{B}}_{\alpha,\beta; \gamma_1,\dots,\gamma_k}$ is only non-zero when $\alpha+\beta-1= \gamma_1+\dots+\gamma_k$, which is essentially the quasi-homogeneity condition.

Consider the infinite system of PDEs on $f = f(t_1,t_2,\dots)$, Eq. \eqref{eq: int sys}:
\begin{equation}\label{eq: Binfty Cauchy}
    \p_\alpha \p_\beta f = \sum_{k \ge 1} \sum_{\gamma_1 + \dots + \gamma_k = \alpha+\beta -1 } R^{\mathrm{B}}_{\alpha,\beta; \gamma_1,\dots,\gamma_k} \p_1 \p_{\gamma_1} f \cdots \p_1 \p_{\gamma_k} f.
\end{equation}
This is a Cauchy-type system of PDEs expressing any second order derivatives of $\p_\alpha\p_\beta f$ via the special second order derivatives $\p_1\p_\bullet f$.

It follows from Proposition~\ref{prop: compatibility from WDVV} that the system of PDEs \eqref{eq: Binfty Cauchy} is compatible.

\begin{theorem}\label{theorem: Binfty hierarchy is BKP}
    The system of PDES \eqref{eq: Binfty Cauchy} coincides with the dispersionless BKP hierarchy \eqref{bkp-dl} after the change of variables $t_k \mapsto t_{2k-1}/(2k-1)$ and the substitution $f \mapsto 2F_0$.
\end{theorem}
\begin{proof}
    Note that $D^{\mathrm{B}}(z) = \left.D(z) \right|_{z^{2k} = 0}$. 
    Denote
    \[
        g^{BKP}(z_1,z_2)\cdot \phi := \frac{\p_1 (2D^{\mathrm{B}}(z_1) - 2D^{\mathrm{B}}(z_2)) \phi(\bt)}{z_1-z_2},
        \quad
        g^{KP}(z_1,z_2) \cdot \psi := \frac{\p_1 (D(z_1) - D(z_2))\psi(\bt)}{z_1-z_2}.
    \]
    We have 
    \[
    g^{BKP}(z_1,z_2)\cdot \phi = \left( \left.g^{KP}(z_1,z_2)\right|_{z_1^{2k} = z_2^{2k} = 0} \right) \cdot 2\phi.
    \]
    Eq.~\eqref{bkp-dl} reads
    \begin{align*}
        4D^{\mathrm{B}}(z_1)D^{\mathrm{B}}(z_2) F_0 &= \log\left(\frac{1 - g^{BKP}(z_1,z_2) \cdot F_0}{1 - g^{BKP}(z_1,-z_2) \cdot F_0} \right)
        \\
        &= \left.\left[\log\left(1 - g^{KP}(z_1,z_2) \cdot 2F_0 \right) - \log\left(1 - g^{KP}(z_1,-z_2) \cdot 2F_0 \right)\right]\right|_{z_1^{2k}=z_2^{2k} = 0}.
    \end{align*}
    Expanding RHS in series, one gets the following identity:
    \begin{align*}
        \sum_{\substack{i,j \ge 1 \\ i,j \not\in 2\ZZ}} &4 \p_i\p_j F_0 z_1^{-i} z_2^{-j}
        \\
        &= 
        2 \sum_{m \ge 1} \frac{(-1)^{m-1}}{m} \sum_{\substack{i,j \ge 1 \\ i,j \not\in 2\ZZ}} z_1^{-i} z_2^{-j} \sum_{\substack{\sum_{k=1}^m\gamma_k= i+j-m \\ \gamma_k \not\in 2\ZZ}} \frac{ij \cdot \widetilde P_{ij}(\gamma_1,\dots,\gamma_m)}{\gamma_1\cdots\gamma_m} \prod_{k=1}^m
        2\p_1\p_{\gamma_k} F_0.
    \end{align*}
    The proof follows now by Corollary~\ref{corollary: Dn enumuerative}.
\end{proof}

\section{\texorpdfstring{$D$}{D} hierarchy}\label{sec:Dhier}

Consider the series of rational numbers $R^{(D_N,1)}$ and $R^{(D_N,2)}$ defined as follows. 
For any fixed $\alpha,\beta \ge 1$ set
% {\color{red}
% \begin{align}
% %     R^{(0,N)}_{\alpha,\beta; \gamma_1,\dots,\gamma_k} &:= \left[t_{N-\gamma_1}\cdot\dots\cdot t_{N-\gamma_k}\right] \frac{\p \B_\alpha^{(N)}}{\p t_\beta},
% %     \\
%     & R^{(D_N,1)}_{\alpha,\beta; \gamma_1,\dots,\gamma_m} := \left[t_{N-\gamma_1}\cdot\dots\cdot t_{N-\gamma_m}\right] \frac{\p^2F_{D_N}}{\p t_\alpha \p t_\beta},
%     \\
%     & R^{(D_N,2)}_{\alpha; \gamma_1,\dots,\gamma_m} := \left[t_N \cdot t_{N-\gamma_1}\cdot\dots\cdot t_{N-\gamma_m}\right] \frac{\p^2F_{D_N}}{\p t_\alpha \p t_N},
% \end{align}
% }
\begin{align*}
    & R^{(D_N)}_{\alpha,\beta; \sigma_1,\dots,\sigma_m} := \left.\frac{1}{m!} \frac{\p^{m+2} F_{D_N}}{\p t_\alpha \p t_\beta \p t_{\bar\sigma_1}\cdot\dots\cdot \p %´
    t_{\bar\sigma_m}} \right|_{\bt =0}, \quad 1 \le \sigma_k \le N,
\end{align*}
where we use the notation $\bar N := N$ and $\bar \kappa = N - \kappa$ for $\alpha < N$.

Recall that dependence of $F_{D_N}$ on the variable $t_N$ is very special (see Section~\ref{section: DN revisited}). We have $\p_\alpha\p_\beta F_{D_N} = \p_\alpha \A_\beta^{(N)} + t_N^2 \p_\alpha \B_\beta$ and $\p_N F_{D_N} = t_N v_1^{(N)}$. The variable $t_N$ has non-zero weight, therefore it follows that for $\alpha,\beta < N$, and $\gamma_k < N$ we have
\begin{align*}
    & R^{(D_N)}_{\alpha,\beta; \gamma_1,\dots,\gamma_m} = \left.\frac{1}{m!} \frac{\p^{m+1} \A_\beta^{(N)}}{\p_\alpha \p t_{N-\gamma_1}\cdot\dots\cdot \p t_{N-\gamma_m}} \right|_{\bt=0},
    \\
    & R^{(D_N)}_{\alpha,N; \gamma_1,\dots,\gamma_m} = \left.\frac{1}{m!} \frac{\p^{m+1} v_1(\bt)}{\p t_\alpha \p t_{N-\gamma_1}\cdot\dots\cdot \p t_{N-\gamma_m}} \right|_{\bt=0}.
\end{align*}
By Theorem~\ref{theorem: Dn stabilization} the following quantities are well-defined
\begin{align*}
    R^{(\mathrm{D},1)}_{\alpha,\beta; \gamma_1,\dots,\gamma_m} &:= R^{(D_{\alpha+\beta-1})}_{\alpha,\beta; \gamma_1,\dots,\gamma_m},
    \\
    R^{(\mathrm{D},2)}_{\alpha; \gamma_1,\dots,\gamma_m} &:= R^{(D_{\alpha+\beta-1})}_{\alpha,N; \gamma_1,\dots,\gamma_m}.
\end{align*}
It follows immediately from the weights counting that $R^{(\mathrm{D},1)}_{\alpha,\beta; \gamma_1,\dots,\gamma_m}$ is only non-zero when $\gamma_1+\dots+\gamma_m = \alpha+\beta-1$ and $R^{(\mathrm{D},2)}_{\alpha; \gamma_1,\dots,\gamma_m}$ is only non-zero when $\gamma_1+\dots+\gamma_m = \alpha-1$.

For a function $f = f(t_0,t_1,t_2,\dots)$ consider the system of PDEs
\begin{align}
    \p_\alpha \p_\beta f &= \sum_{m \ge 1}\sum_{\gamma_1,\dots,\gamma_m} R^{(\mathrm{D},1)}_{\alpha,\beta; \gamma_1,\dots,\gamma_m} \prod_{k=1}^m \p_1 \p_{\gamma_a} f 
%     + t_0^2 \cdot \sum_k\sum_{\gamma_\bullet} R^{(0)}_{\alpha,\beta; \gamma_1,\dots,\gamma_k} \prod_{a=1}^k \p_1 \p_{\gamma_a} f
    \label{eq: DN flow A}
    \\
    \p_0 \p_\alpha f &= \p_0 \p_1 f \cdot \sum_{m\ge 1} \sum_{\gamma_1,\dots,\gamma_m} R^{(\mathrm{D},2)}_{\alpha; \gamma_1,\dots,\gamma_m} \prod_{k=1}^m \p_1 \p_{\gamma_a} f,
    \label{eq: DN flow B}
\end{align}
for all $\alpha,\beta \ge 2$.

Denote $p_k := \p_1\p_k f$. 
The first PDE's of the system above read
\begin{align*}
    \p_2 \p_2 f &= \frac{1}{12} p_1^3 -\frac{1}{2} p_2 p_1 + p_3,
    \\
    \p_2 \p_3 f &= -\frac{1}{2} p_2^2 + \frac{1}{4} p_1^2 p_2  - \frac{1}{2} p_1 p_3 + p_4 ,
    \\
    \p_2 \p_4 f &= \frac{1}{4} p_1 p_2^2 + \frac{1}{4} p_1^2 p_3  -\frac{1}{2} p_1 p_4  - p_2 p_3 + p_5 ,
    \\
    \p_3 \p_3 f &= \frac{1}{80} p_1^5 -\frac{1}{8} p_1^3 p_2  +\frac{3}{4} p_1 p_2^2  +\frac{1}{4} p_1^2 p_3  -\frac{1}{2}p_1 p_4 -\frac{3}{2}p_2 p_3 + p_5
\end{align*}
and
\begin{align*}
    \p_0 \p_2 f &= \frac{1}{2}p_1,
    \\
    \p_0 \p_3 f &= \frac{1}{2}p_1^2 + \frac{1}{2} p_2,
    \\
    \p_0 \p_4 f &= \frac{1}{8} p_1^3 + \frac{1}{2} p_1 p_2 + \frac{1}{2} p_3,
    \\
    \p_0 \p_5 f &= \frac{1}{16} p_1^4 + \frac{3}{8} p_1^2p_2 + \frac{1}{4} p_2^2 + \frac{1}{2}p_2p_3 + \frac{1}{2}p_4.
\end{align*}

%This follows from 
Proposition~\ref{prop: compatibility from WDVV} implies that this system of PDEs is compatible.

\subsection{\texorpdfstring{2--component}{2-component} BKP}
Consider the function $\tau  = \tau(\bt,\bar\bt)$, for $\bt = \lbrace t_1, t_3, t_5, \dots \rbrace$ and $\bar\bt = \lbrace \bar t_1, \bar t_3, \bar t_5, \dots \rbrace$ being two sets of independent variables. 

Together with the operators $D^{\mathrm{B}}(z)$ and $\Delta^{\mathrm{B}}(z)$ from Section~\ref{section: 1-BKP} we will need two similar operators acting on the additional set of variables $\bar \bt$:
\[
    \bar D^{\mathrm{B}}(z) := \sum_{n \ge 0} \frac{z^{-2n-1}}{2n+1} \bar\p_{2n+1},
    \quad
    \bar \Delta^{\mathrm{B}}(z) := \frac{\exp(2 \hbar \bar D^{\mathrm{B}}(z)) - 1}{\hbar},
\]
where $\bar\p_k := \p / \p t_k$. Denote 
\[
    \tau(\bt \pm 2[z^{-1}], \bar\bt) := e^{\pm 2 \hbar D(z)} \cdot \tau(\bt,\bar\bt), \quad \tau(\bt, \bar\bt \pm 2[z^{-1}]) := e^{\pm 2 \hbar\bar D(z)} \cdot \tau(\bt,\bar\bt).
\]
Set also $\xi(\bt,z) := \sum_{n\ge 0} t_{2n+1}z^{2n+1}$.

The 2--component BKP hierarchy is the following equation
\begin{align*}
    &\res\left( e^{\xi(\bt'-\bt,z)} \tau(\bt' - 2[z^{-1}],\bar\bt') \tau(\bt + 2[z^{-1}],\bar\bt)\frac{dz}{z}\right)
    \\
    &\quad\quad =
    \res\left( e^{\xi(\bar\bt'-\bar\bt,z)} \tau(\bt',\bar\bt'- 2[z^{-1}]) \tau(\bt ,\bar\bt+ 2[z^{-1}])\frac{dz}{z}\right).
\end{align*}
This equation coincides with the BKP hierarchy by putting $\bt' = \bt$.

The set of Fay identities equivalent to the equation above were derived in \cite[Eqq.4.9 - 4.12]{T}. They read
\begin{align}
    &\left(z_1 + z_2 - \hbar \p_1 \Delta^{\mathrm{B}}(z_1)\Delta^{\mathrm{B}}(z_2) F  - \p_1(\Delta^{\mathrm{B}}(z_1)F + \Delta^{\mathrm{B}}(z_2)F)\right)
    \exp(\Delta^{\mathrm{B}}(z_1)\Delta^{\mathrm{B}}(z_2) F)
    \tag{2BKP-1}\label{2bkp-1}
     \\
     &\quad\quad\quad = \frac{z_1+z_2}{z_1-z_2} \left( z_1-z_2 - \p_1(\Delta^{\mathrm{B}}(z_1)F - \Delta^{\mathrm{B}}(z_2)F) \right).
     \notag
     \\
    &\left(z_1 + z_2 - \hbar \bar\p_1 \bar\Delta^{\mathrm{B}}(z_1)\bar\Delta^{\mathrm{B}}(z_2) F  - \bar\p_1(\Delta^{\mathrm{B}}(z_1)F + \bar\Delta^{\mathrm{B}}(z_2)F)\right)
    \exp(\bar\Delta^{\mathrm{B}}(z_1)\bar\Delta^{\mathrm{B}}(z_2) F)
    \tag{2BKP-2}\label{2bkp-2}
     \\
     &\quad\quad\quad = \frac{z_1+z_2}{z_1-z_2} \left( z_1-z_2 - \bar\p_1(\bar\Delta^{\mathrm{B}}(z_1)F - \bar\Delta^{\mathrm{B}}(z_2)F) \right).
     \notag
     \\
     &\left(z_1  - \hbar \p_1 \Delta^{\mathrm{B}}(z_1)\bar\Delta^{\mathrm{B}}(z_2) F  - \p_1(\Delta^{\mathrm{B}}(z_1)F + \bar\Delta^{\mathrm{B}}(z_2)F)\right)
    \exp(\Delta^{\mathrm{B}}(z_1)\Delta^{\mathrm{B}}(z_2) F)
     \tag{2BKP-3}\label{2bkp-3}
     \\
     &\quad\quad\quad = z_1 - \p_1(\Delta(z_1)-\bar\Delta(z_2))F,
     \notag
     \\
    &\left(z_2 - \hbar \bar \p_1 \Delta^{\mathrm{B}}(z_1)\bar\Delta^{\mathrm{B}}(z_2) F  - \bar \p_1(\Delta^{\mathrm{B}}(z_1)F + \bar\Delta^{\mathrm{B}}(z_2)F)\right)
    \exp(\Delta^{\mathrm{B}}(z_1)\Delta^{\mathrm{B}}(z_2) F)
     \tag{2BKP-4}\label{2bkp-4}
     \\
     &\quad\quad\quad = z_2 - \bar \p_1(\bar\Delta(z_2)-\Delta(z_1))F.
     \notag
\end{align}
One notes immediately that \eqref{2bkp-1} coincides with \eqref{2bkp-2} after the interchange $\bt$ and $\bt'$.

In the case of 2-component BKP hierarchy we do not have any result similar to Proposition~\ref{prop: Fay by NZ} above. However it is not hard to derive the dispersionless form of the Fay-type identities above. 

Assume $\tau(\bt,\bar\bt) = \hbar^{2}F$ with $F = \sum_{g\ge 0}\hbar^gF_g(\bt,\bar\bt)$. Dispersionless limit of the Fay-type identities above reads:
\begin{align}
    \left(1 - \frac{ \p_1 (2D^{\mathrm{B}}(z_1) + 2D^{\mathrm{B}}(z_2))F_0}{z_1 + z_2} \right) &
    \ e^{2D^{\mathrm{B}}(z_1)\cdot2D^{\mathrm{B}}(z_2) F_0}
    \tag{2BKP-1dl}\label{2bkp-1dl}
     \\
     & 
     = \left(1 - \frac{\p_1 (2D^{\mathrm{B}}(z_1) - 2D^{\mathrm{B}}(z_2))F_0}{z_1-z_2} \right),
     \notag
     \\
     \left(1 - \frac{\bar\p_1 (2\bar D^{\mathrm{B}}(z_1) + 2\bar D^{\mathrm{B}}(z_2))F_0}{z_1+z_2} \right) &
     \ e^{2\bar D^{\mathrm{B}}(z_1)\cdot 2\bar D^{\mathrm{B}}(z_2) F_0}
    \tag{2BKP-2dl}\label{2bkp-2dl}
     \\
     &= \left(1 - \frac{\bar\p_1 (2\bar D^{\mathrm{B}}(z_1) - 2\bar D^{\mathrm{B}}(z_2))F_0}{z_1-z_2} \right),
     \notag
     \\
     \left(z_1 - \p_1 (2D^{\mathrm{B}}(z_1) +2\bar D^{\mathrm{B}}(z_2))F_0 \right) &
    \ e^{2D^{\mathrm{B}}(z_1)\cdot 2\bar D^{\mathrm{B}}(z_2)F_0}
     \tag{2BKP-3dl}\label{2bkp-3dl}
     \\
     &= z_1 - \p_1(2D^{\mathrm{B}}(z_1)- 2\bar D^{\mathrm{B}}(z_2))F_0,
     \notag
     \\
     \left(z_2 - \bar\p_1 (2D^{\mathrm{B}}(z_1) +2\bar D^{\mathrm{B}}(z_2))F_0 \right) & 
     \ e^{2D^{\mathrm{B}}(z_1)\cdot 2\bar D^{\mathrm{B}}(z_2)F_0}
     \tag{2BKP-4dl}\label{2bkp-4dl}
     \\
     &= z_2 - \bar\p_1(2\bar D^{\mathrm{B}}(z_2)- 2 D^{\mathrm{B}}(z_1))F_0.
     \notag
\end{align}

\subsection{Identification}
The connection between our Cauchy-type hierarchy of $D_\infty$ Frobenius manifold and 2-component BKP hierarchy requires taking a certain reduction in $\bar\bt$ variable. 

In particular, we assume the Fay-type identities above to be reduced by setting
\[
    \bar D^{\mathrm{B}}(z) = z^{-1}\bar\p_1,
\]
and assuming all operators $\bar D^{\mathrm{B}}(z)$ only in the $1$-st order in $z^{-1}$.

That is, the solution $F_0$ of the reduced hierarchy should satisfy Eq.~\eqref{2bkp-1dl} and also
\begin{align}\label{red2bkp-3dl}
     D^{\mathrm{B}}(z_1)\cdot \bar\p_1 2F_0
     = \sum_{m \ge 1} \bar\p_1\p_1 2 F_0 \cdot z_1^{-m}
     \left(\p_1 D^{\mathrm{B}}(z_1) 2F_0 \right)^{m-1}.
\end{align}
while Eq.~\eqref{2bkp-2dl} and Eq.~\eqref{2bkp-4dl} hold trivially in this reduction.

\begin{theorem}
    The system of PDEs \eqref{eq: DN flow A}--\eqref{eq: DN flow B} 
    coincides with the discussed above reduction of the 2--component BKP hierarchy after the change of variables $t_k \mapsto t_{2k-1}/(2k-1)$, $t_0 \mapsto \bar t_1$ and substitution $f \mapsto 2F_0$.
\end{theorem}
\begin{proof}
    Note that Eq.~\eqref{2bkp-1dl} coincides with Eq.~\eqref{bkp-dl}. It follows immediately from Corollary~\ref{corollary: Dn enumuerative} and Theorem~\ref{theorem: Binfty hierarchy is BKP} that the system~\eqref{eq: DN flow A} coincides after the discussed substitution with Eq.~\eqref{2bkp-1dl}.
    
    We show that Eq.~\eqref{eq: DN flow B} coincides with Eq.~\eqref{red2bkp-3dl} after the discussed substitution.
    This is essentially a question about the rational numbers $R^{(\mathrm{D},2)}_{\alpha; \gamma_1,\dots,\gamma_k}$.
    
    Applying the discussed substitution to Eq.\eqref{red2bkp-3dl} we have
    \begin{align*}
     &\sum_{a\ge 1} z_1^{-(2a-1)} \p_a\bar\p_1 f
     = \sum_{m \ge 1} \p_1\bar\p_1 f \cdot z_1^{-m}
     \left(\sum_{b\ge 1} z_1^{-(2b-1)} \p_1\p_b f \right)^{m-1}
     \\
     \Leftrightarrow
     &\quad \p_a\bar\p_1 f = \p_1\bar\p_1 f \sum_{m \ge 0} \sum_{\gamma_\bullet} \p_1\p_{\gamma_1}f \cdots \p_1\p_{\gamma_m}f,
    \end{align*}
    where the last summation is taken over all $2a-1 = m +1 + \sum_{k=1}^m (2\gamma_k-1)$ that hold if and only if $a = 1 + \sum_{k=1}^m \gamma_k$. This coincides exactly with the expansion of $\p v_1 / \p t_a$ in the $s_{\gamma_k} = t_{N-\gamma_k}$ variables, c.f. Eq.\eqref{eq: Dn essential coordinate via flat}. This completes the proof.
\end{proof}

\end{document}